%% file: draft_jour.tex
\documentclass[12pt, draftcls, onecolumn]{IEEEtran}

\ifCLASSOPTIONcompsoc
\else
\fi
\ifCLASSINFOpdf
\else
\fi

\makeatletter
\def\ps@headings{%
\def\@oddhead{\mbox{}\scriptsize\rightmark \hfil \thepage}%
\def\@evenhead{\scriptsize\thepage \hfil \leftmark\mbox{}}%
\def\@oddfoot{}%
\def\@evenfoot{}}
\makeatother 
\usepackage{cite}
\usepackage{fancyhdr}
\usepackage{amssymb}
\ifCLASSINFOpdf
\else
  \usepackage[dvips]{graphicx}
\fi
\usepackage[cmex10]{amsmath}
\usepackage{algorithmic}
\usepackage[tight,footnotesize]{subfigure}
\makeatletter
\newtheorem{lem}{Lemma}
\newtheorem{thm}{Theorem}
\newtheorem{condition}{Condition}

\newcommand{\Rmnum}[1]{\expandafter\@slowromancap\romannumeral #1@}
\makeatother

\begin{document}
%
% paper title
% can use linebreaks \\ within to get better formatting as desired
\title{Channel Estimation for Opportunistic Spectrum Access: Uniform and Random Sensing}
\author{Quanquan~Liang,~
        Mingyan~Liu,~\IEEEmembership{Member,~IEEE,}
\IEEEcompsocitemizethanks{\IEEEcompsocthanksitem Q. Liang and M. Liu are with Department of Electrical Engineering and Computer Science, University of Michigan, Ann Arbor, MI 48109-2122.
Email: liangq@umich.edu, mingyan@eecs.umich.edu
}\thanks{An earlier version of this paper appeared at the Information Theory and Application Workshop (ITA), UC San Diego, CA, February 2010.}}

\IEEEcompsoctitleabstractindextext{%
\begin{abstract}
The knowledge of channel statistics can be very helpful in making sound opportunistic spectrum access decisions.  It is therefore desirable to be able to efficiently and accurately estimate channel statistics.  In this paper we study the problem of optimally placing sensing times over a time window so as to get the best estimate on the parameters of an on-off renewal channel. We are particularly interested in a sparse sensing regime with a small number of samples relative to the time window size.
Using Fisher information as a measure, we analytically derive the best and worst sensing sequences under a sparsity condition.  We also present a way to derive the best/worst sequences without this condition using a dynamic programming approach.  In both cases the worst turns out to be the uniform sensing sequence, where sensing times are evenly spaced within the window.  With these results we argue that without a priori knowledge, a robust sensing strategy should be a randomized strategy.  We then compare different random schemes using a family of distributions generated by the circular $\beta$  ensemble, and propose an adaptive sensing scheme to effectively track time-varying channel parameters.  We further discuss the applicability of compressive sensing for this problem.
\end{abstract}

}

% make the title area
\maketitle

%\thispagestyle{fancy}
%\fancyhead{}
%\lhead{\today}
%\chead{}
%\rhead{DRAFT:  Version \Rmnum{4}}
%\lfoot{}
%\cfoot{}   %current page number
%\rfoot{}
%\renewcommand{\headrulewidth}{0pt}
%\renewcommand{\footrulewidth}{0pt}

\IEEEdisplaynotcompsoctitleabstractindextext
% \IEEEdisplaynotcompsoctitleabstractindextext has no effect when using
% compsoc under a non-conference mode.

\IEEEpeerreviewmaketitle

\input{intro}

\input{model}

\input{estimation}

\input{sensing}

\input{random}

\input{variation}
\input{cs}

\section{Conclusion} \label{sec:conclusion}
In this paper we studied sensing schemes for a channel estimation problem under a sparsity condition.  Using Fisher information as a performance measure, we derived the best and worst sensing sequences both with and without the sparsity condition.  These sequences, while not exactly implementable, provide significant insights as well as useful benchmarks.  We then examined the performance of random sensing schemes, by comparing a family of distributions generated by the circular $\beta$ ensemble. %The best sampling strategy which is obtained by optimize the Fisher information provides us a benchmark to test these random sampling schemes.
Using these insights, an adaptive random sensing scheme was proposed to effectively track time-varying channel parameters.  We also discuss the applicability of compressive sensing in this context.
\appendices
\input{proof}

\ifCLASSOPTIONcaptionsoff
  \newpage
\fi

\end{document}

%% file: intro.tex
\section{Introduction}\label{sec:intro}

\IEEEPARstart{R}{ecent} advances in software defined radio and cognitive radio \cite{Haykin} have given wireless devices greater ability and opportunity to dynamically access spectrum, thereby potentially significantly improving spectrum efficiency and user performance \cite{Challapali,Akyildiz}.  To be able to fully utilize spectrum availability (either as a secondary user seeking opportunities of idle periods in the presence of primary users, or as one of many peer users in a multi-user system seeking channels with the best condition), a key enabling ingredient in dynamic spectrum access is high quality channel sensing that allows the user to obtain accurate real-time information on the condition of wireless channels.

Spectrum sensing is often studied in two contexts: at the physical layer and at the MAC layer.  Physical layer spectrum sensing typically focuses on the {\em detection} of instantaneous primary user signals.  Several detection methods, such as matched filter detection, energy detection and feature detection, have been proposed for cognitive radios \cite{Cabric}. MAC layer spectrum sensing \cite{kim1,Zhao} is more of a resource allocation issue, where we are concerned with the {\em scheduling} problem of when to sense the channel and the {\em estimation} problem of extracting statistical properties of the random variation in the channel, assuming that when we decide to sense the physical layer can provide sufficiently accurate results on instantaneous channel availability. Such channel statistics can be very helpful in making good channel access decisions, and most studies on opportunistic spectrum access assume such knowledge.

In this paper we focus on the scheduling of channel sensing and study the effect different scheduling algorithms have on the accuracy of the resulting estimate we obtain on channel parameters.  In particular, we are interested in the {\em sparse sensing/sampling} regime where we can use only a limited number of measurements over a given period of time.
The goal is to decide how these limited number of measurements should be scheduled so as to minimize the estimation error within the maximum likelihood (ML) estimator framework. Throughout the paper the terms {\em sensing} and {\em sampling} will be used interchangeably.

MAC layer channel estimation within the context of cognitive radios has been studied in recent years.  Below we review those most relevant to the present paper.  Kim and Shin \cite{kim1} introduced a ML estimator for renewal channels using a uniform sampling/sensing scheme where samples of the channel are taken at regular time intervals.  A more accurate, but also much more computationally costly Bayesian estimator was introduced in \cite{kim2}, again based on uniform sensing. \cite{Crowncom} analyzed the relationship between estimation accuracy, number of samples taken and the channel state transition probabilities by using the sampling and estimation framework of \cite{kim1} and focusing on Markovian channels. \cite{HMM} proposed a Hidden Markov Model (HMM) based channel status predictor using reinforcement learning techniques.  This predictor predicts next channel state based on past information obtained through uniformly sampling the channel. \cite{Wavelet} presented a channel estimation technique based on wavelet transform followed by filtering.
This method relies on dense sampling of the channel.

In most of the above cited work, the focus is on the estimation problem given (sufficiently dense) uniform sampling of the channel, i.e., with equal time periods between successive samples.  This scheme will be referred to as {\em uniform sensing} in the remainder of this paper.  By contrast, sampling schemes where time intervals between successive samples are drawn from a certain probability distribution will be referred to as {\em random sensing} throughout the paper.
We observe that due to constraints on time, energy, memory and other resources, a user may wish to perform channel sensing at much lower frequencies while still hoping for good estimates. This could be relevant for instance in cases where a user wants to track the channel condition in between active data communication, or where a user needs to track a large number of different channels. 
It is this sparse sampling scenario that we will focus on in this study, and the goal is to judiciously schedule these limited number of samples.

Our main contributions are summarized as follows.
\begin{itemize}
\item We demonstrate that when sampling is done sparsely, random sensing significantly outperforms uniform sensing. 
\item In the special case of exponentially distributed on/off durations, we derive tight lower and upper bounds on the Fisher information under a sparsity condition, while obtaining the best and worst possible sampling schemes measured by the Fisher information.  We show that uniform sensing is the {\em worst} one can do; any deviation from it improves the estimation accuracy. %We also give the {\em best} sampling sequence under the same sparsity condition.
\item We present a dynamic programming approach to obtain the best and worst sampling sequences in the more general case without the sparsity condition. 
\item We show that under the same channel statistics and the same average sampling interval (or frequency), a random sensing scheme affects the estimation accuracy through the higher-order central moments of the sampling intervals, and use the circular $\beta$ ensemble to study a family of distributions. %In particular, in this case using exponentially distributed sampling intervals results in much better estimation performance compared to others.
\item We present an {\em adaptive random sensing} scheme that can very effectively track time-varying channel parameters, and is shown to outperform its counterpart using uniform sensing.
\end{itemize}

The remainder of this paper is organized as follows: Section \ref{sec:model} presents the channel models and Section \ref{sec:estimation} gives the detail of the ML estimator.  Then in Section \ref{sec:sensing} we present how the sampling scheme affects the estimation performance; the best and worst sensing sequences with and without a sparse sampling condition are obtained.
%Furthermore the optimal/worst sampling strategy without any constraint are also given by dynamic programming.
In Section \ref{sec:diff_random} we use a family of distributions generated by the circular $\beta$ ensemble to examine different random sampling schemes. Section \ref{sec:tracking} presents an adaptive random sensing scheme, and Section \ref{sec:CS} discusses the applicability of compressive sensing in this problem.  Section \ref{sec:conclusion} concludes the paper.

%% file: model.tex
\section{The Channel Model} \label{sec:model}

In this paper we will limit our attention to MAC layer spectrum sensing as mentioned in the introduction.  Within this context, the channel state perceived by a secondary user is represented by a binary random variable. % \footnote{This type of two-state model is not limited to MAC layer sensing studies as the references show, but is a particularly widely used simplification in this setting.}.
This is a model commonly used in a large volume of literature, from channel estimation (e.g., \cite{kim1,Crowncom}) to opportunistic spectrum access (e.g., \cite{Zhao}) to spectrum measurement (e.g., \cite{mobicom09}).
Specifically, let $Z(t)$ denote the state of the channel at time $t$, such that
\begin{displaymath}
\left\{ \begin{array}{ll}
Z(t)=1 &\textrm{if the channel is sensed busy at time $t$} ~,\\
Z(t)=0 &\textrm{otherwise .} \end{array} \right.
\end{displaymath}

The advantage of such a model is its simplicity and tractability in many instances.  The weakness lies in the fact that the actual energy present or detected in the channel is hardly binary.  The raw channel measurement data will have to go through a binary hypothesis test (e.g., via thresholding) to be reduced to the above form, a process that comes with probabilities of error.  Consequently, the channel is sensed to be in either state with a detection probability and a false alarm probability.

In this paper our focus is on extracting and {\em estimating} essential statistics given a sequence of measured channel states (0s and 1s) rather than the binary {\em detection} of channel state (deciding between 0 and 1 given the energy reading).  For this purpose, we will assume that the channel state measurements are error-free.  If we have side information on what the detection and false alarm probabilities are, then the estimation results may be adjusted accordingly to utilize such knowledge.

The channel state process $Z(t)$ is assumed to be a continuous-time alternating renewal process, alternating between on/busy (state ``1'') and off/idle (state ``0''), an illustration is given in Figure \ref{channel model}. Typically, it is assumed that a secondary user can utilize the channel only when it is sensed to be in the off states (i.e., when the channel is idle or the primary user is absent). When the channel state transitions to the on state, the secondary user is required to vacate the channel so as not to interfere with the primary user (also referred to as the spectrum underlay paradigm, see e.g., \cite{underlay}).

This random process is completely defined by two probability density functions $f_{1}(t)$ and $f_{0}(t)$, $t>0$, i.e., the probability distribution of the sojourn times of the on periods (denoted by the random variable $T_{1}$) and the off periods (denoted by the random variable $T_{0}$), respectively. The channel utilization $u$ is defined as
\begin{equation}
\label{def:u}
u=\frac{E[T_{1}]}{E[T_{1}]+E[T_{0}]},
\end{equation}
which is also the average fraction of time the channel is occupied or busy.
By the definition of a renewal process,
 $T_{1}$ and $T_{0}$ are independent and all on (off) periods are independently and identically distributed. It's worth pointing out that the widely used Gilbert-Elliot model (a two-state Markov chain) is a special case of the alternating renewal process where the on (off) periods are exponentially (in the case of continuous time) or geometrically (in the case of discrete time) distributed.

\begin{figure}[htb]
  % Requires \usepackage{graphicx}
 \centering
  \includegraphics[scale=0.26]{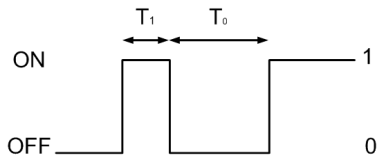}\\
 \centering
 \parbox{11cm}{\caption{Channel model: alternating renewal process with on and off states}\label{channel model}}
\end{figure}

%% file: estimation.tex
\section{Maximum Likelihood (ML) Based Channel Estimation} \label{sec:estimation}

We proceed to describe the maximum likelihood (ML) estimator \cite{ML_defination} we will use to estimate channel parameters from a sequence of channel state observations.

Recall that the channel state is assumed to follow an alternating renewal process.  Such a process is completely characterized by the set of conditional probabilities
%Firstly we introduce the channel state transition probability over an interval $\Delta t$ denote by
$P_{ij}(\Delta t)$, $i, j\in \{0, 1\}$, $\Delta t\geq 0$, defined as the probability that given $i$ was observed $\Delta t$ time units ago, $j$ is now observed.  This quantity is also commonly known as the semi-Markov kernel of an alternating renewal process \cite{renewal}.
%This probability is very useful in MAC layer sensing since it is the probability that channel would be idle or busy at a certain time $\Delta t$ based on the previous samples. And it is also used in our likelihood function introduced later.
Assuming the process is in equilibrium, standard results from renewal theory \cite{renewal} suggest the following Laplace transforms of the above transition probabilities:
%Considering the process which has been in operation for a long time and using standard result from renewal theory \cite{renewal},  the Laplace transform of the above transition probability is given as follows for our on/off renewal process:\\
\begin{equation}
\begin{split}
\label{equ1}
&P^*_{00}(s)=\frac{1}{s}-\frac{\left\{1-f_{1}^*(s)\right\}\left\{1-f_{0}^*(s)\right\}}{E[T_{0}]s^2\left\{1-f_{1}^*(s)f_{0}^*(s)\right\}} ~,\\
&P^*_{01}(s)=\frac{\left\{1-f_{1}^*(s)\right\}\left\{1-f_{0}^*(s)\right\}}{E[T_{0}]s^2\left\{1-f_{1}^*(s)f_{0}^*(s)\right\}} ~,\\
&P^*_{10}(s)=\frac{\left\{1-f_{1}^*(s)\right\}\left\{1-f_{0}^*(s)\right\}}{E[T_{1}]s^2\left\{1-f_{1}^*(s)f_{0}^*(s)\right\}} ~,\\
&P^*_{11}(s)=\frac{1}{s}-\frac{\left\{1-f_{1}^*(s)\right\}\left\{1-f_{0}^*(s)\right\}}{E[T_{1}]s^2\left\{1-f_{1}^*(s)f_{0}^*(s)\right\}} ~, \\
\end{split}
\end{equation}
where $f_{1}^*(s)$ and $f_{0}^*(s)$ are the Laplace transforms of $f_{1}(t)$ and $f_{0}(t)$, respectively.
We see %from this set of equations that the channel statistics is
that these are completely defined by the probability density functions $f_{1}(t)$ and $f_{0}(t)$.
The above set of equations are very useful in recovering the time-domain expressions of the semi-Markov kernel (often times this is the only viable method).
For example, in the special case where the channel has exponentially distributed on/off periods, we have
\begin{equation}
\label{exp_distribution}
\left\{ \begin{array}{ll}
f_{1}(t)=\theta_{1}e^{-\theta_{1}t} \\
f_{0}(t)=\theta_{0}e^{-\theta_{0}t} ~.
\end{array}
\right.
\end{equation}
Their corresponding Laplace transforms and expectations are
\begin{eqnarray*}
\left\{ \begin{array}{ll}
f_{1}^*(s)=\theta_{1}/(s+\theta_{1}) \\
f_{0}^*(s)=\theta_{0}/(s+\theta_{0}) ~,
\end{array}
\right.
\quad \quad \quad
\left\{ \begin{array}{ll}
E[T_{1}]=1/\theta_{1} \\
E[T_{0}]=1/\theta_{0} ~.
\end{array}
\right.
\end{eqnarray*}
Substituting the above expressions into  (\ref{equ1}) followed by an inverse Laplace transform we get the state transition probability as follows:
\begin{equation}
\label{exp}
P_{ij}(\Delta t)=u^{j}(1-u)^{1-j}+(-1)^{j+i}u^{1-i}(1-u)^{i}e^{-(\theta_{0}+\theta_{1})\Delta t} ~,
\end{equation}
where $u=\frac{E[T_{1}]}{E[T_{1}]+E[T_{0}]}$, as defined earlier.

In this paper we consider the following estimation problem.  Assume that the on/off periods are given by certain known distribution functions $f_0(t)$ and $f_1(t)$ but with unknown parameters.   Suppose we obtain $m$ samples $\{z_{1}, z_{2}, \cdots, z_{m}\}$, taken at sampling times $\{t_{1},t_{2},\cdots,t_{m}\}$, respectively.  We wish to use these samples to estimate the unknown parameters.

First note that the channel utilization factor $u$ %(average fraction of time the channel is busy)
can be estimated through the sample mean of the $m$ measurements as follows
\begin{equation}
\label{eqn:channel_utilization}
\hat{u}=\frac{1}{m}\sum_{i=1}^{m}z_{i}~.
\end{equation}
%which is an unbiased estimator for the true quantity $u$.

Let $\bar{\theta}$ be the unknown parameters of the on/off distributions: $\bar{\theta}=\{\bar\theta_{1}, \bar\theta_{0}\}$.  Note that in general $\bar\theta_1$ and $\bar\theta_0$ are vectors themselves.  Then the likelihood function is given by
\begin{eqnarray}
L({\overline{\theta}})&=& Pr\{\overline{Z};\overline{\theta}\}\nonumber \\
&=& Pr\{Z_{t_{m}}=z_{m}, Z_{t_{m-1}}=z_{m-1}, Z_{t_{m-2}}=z_{m-2},\dots, Z_{t_{1}}=z_{1};\overline{\theta}\} ~.
\end{eqnarray}

The idea of ML estimation is to find the value of $\overline{\theta}$ that maximizes the log likelihood function $\ln L(\overline{\theta})$, i.e., the estimate $\hat{\overline{\theta}}$ is such that $\frac{\partial{\ln L(\overline{\theta})}}{\partial{\overline{\theta}}} |_{\hat{\overline{\theta}}} = 0$.  This method has been used extensively in the literature \cite{ML1,ML2,ML3,MLchannel1,MLchannel2}.  For a fixed set of data and underlying probability model, the ML estimator selects the parameter value that makes the data ``most likely'' among all possible choices.  Under certain (fairly weak) regularity conditions the ML estimator is asymptotically optimal \cite{ML}.

The question we wish to investigate is what impact the selection of the sampling time sequence $\{t_{1}, t_{2}, \cdots, t_{m}\}$ has on the performance of this estimator, given a limited number of samples $m$.
Specifically, we question whether random sampling is a better way of sensing the channel than uniform sampling where the measurement samples are taken at regular time intervals.  %As we mentioned in the introduction, most existing literature focuses on the latter method.

For the remainder of our analysis we will limit our attention to the case where the channel on/off durations are given by exponential distributions.  This is for both mathematical tractability and simplicity of presentation.  We explore other distributions in our numerical experiments.

Since the exponential distribution is defined by a single parameter, we have now $\bar{\theta}=\{\theta_{1}, \theta_{0}\}$, where $\theta_1$ and $\theta_0$ are the two unknown scalar parameters of the on and off exponential distributions, respectively.  Using the memoryless property, the likelihood function becomes
\begin{eqnarray}
\label{likfun}
L({\overline{\theta}})&=&Pr\{\overline{Z};\overline{\theta}\} \nonumber\\
&=& Pr\{Z_{t_{1}}=z_{1};\overline{\theta}\}\cdot\prod_{i=2}^{m}Pr\{Z_{t_{i}}=z_{i}|Z_{t_{i-1}}=z_{i-1};\overline{\theta}\} \nonumber \\
&=& Pr\{Z_{t_{1}}=z_{1};\overline{\theta}\}\cdot\prod_{i=2}^{m}P_{z_{i-1}z_{i}}(\Delta t_{i};\overline{\theta})~.
\end{eqnarray}
%by denoting the sampling intervals as $\{\Delta t_{2},\Delta t_{3},\ldots,\Delta t_{m}\}$,
where $\Delta t_{i}=t_{i}-t_{i-1}$.  The first quantity on the right is taken to be%we have
\begin{equation}
Pr\{z_{t_{1}}=z_{1};\overline{\theta}\} = u^{z_{1}}(1-u)^{1-z_{1}} ~. \label{eqn:first-term-u}
\end{equation}
That is, the probability of finding the channel in a particular state (LHS of Eqn  (\ref{eqn:first-term-u})) is taken to be the stationary distribution given by the RHS.  This choice is justified by assuming that the channel is in equilibrium.

The second quantity  $P_{z_{i-1}z_{i}}(\Delta t_{i}; \bar\theta)$ is given in  Eqn (\ref{exp}).
Combining these two quantities, we have
\begin{equation}
\label{like_full}
\begin{split}
L(\theta_{0}, \theta_{1})=& L(\bar\theta) \\
=&u^{z_{1}}(1-u)^{1-z_{1}}\prod_{i=2}^{m}\big(u^{z_{i}}(1-u)^{1-z_{i}}+(-1)^{z_{i}+z_{i-1}}u^{1-z_{i-1}}(1-u)^{z_{i-1}}e^{-(\theta_{0}+\theta_{1})\Delta t_{i}}\big) ~.
\end{split}
\end{equation}

The estimates for the parameters are found by solving
\begin{equation}
\left\{ \begin{array}{ll}
\frac{\partial{\ln L(\theta_{0}, \theta_{1})}}{\partial\theta_{0}}=0 \\
\frac{\partial{\ln L(\theta_{0}, \theta_{1})}}{\partial\theta_{1}}=0 ~.
%\frac{\theta_{0}}{\theta_{0}+\theta_{1}}=u ~. \nonumber
\end{array}
\label{frac}
\right.
\end{equation}

Technically, to get the estimates for both $\theta_0$ and $\theta_1$ one needs to solve the above two equations simultaneously.  This however proves to be computationally complex and analytically intractable.  Instead, we adopt the following estimation procedure.  We first estimate $u$ using  Eqn (\ref{eqn:channel_utilization}), and take $\theta_1 =\frac{(1-u)\theta_{0}}{u}$.  Due to the exponential assumption, it can be shown that this estimate of $u$ is unbiased regardless of the sequence $\{t_1, \cdots, t_m\}$ as long as it is determined offline.
The likelihood function (\ref{like_full}) can then be re-written as
%{\bf We can find that $\theta_{0}$ and $\theta_{1}$ could interchange in (\ref{like_full}), hence we solve (\ref{frac}) recurring to $\frac{\theta_{0}}{\theta_{0}+\theta_{1}}=u$.}  Substitute $\theta_{1}=\frac{(1-u)\theta_{0}}{u}$,  we can rewrite the likelihood function as:
\begin{equation}
\begin{split}
\label{eqn:likelihood-theta0}
L(\theta_{0})=&u^{z_{1}}(1-u)^{1-z_{1}}\prod_{i=2}^{m}\big(u^{z_{i}}(1-u)^{1-z_{i}}+(-1)^{z_{i}+z_{i-1}}u^{1-z_{i-1}}(1-u)^{z_{i-1}}e^{-\theta_{0}\Delta t_{i}/u}\big) ~.
\end{split}
\end{equation}
%The estimation of $u$ can be obtained by (\ref{channel_utilization}) and
The estimation of $\theta_{0}$ is then derived by solving the equation $\frac{\partial{\ln L(\theta_{0})}}{\partial\theta_{0}}=0$.
%To estimate $\theta_{1}$ we will simply use $\hat{\theta}_{1} = \frac{ (1-\hat{u})\hat{\theta}_{0}}{\hat{u}}$.

In our analysis, we will use this procedure by treating $u$ as a known constant and solely focus on the estimation of $\theta_0$, with the understanding that $u$ is separately and unbiasedly
estimated, and once we have the estimate for $\theta_0$ we have the estimate for $\theta_1$.   It has to be noted that this procedure is in general not equivalent to solving (\ref{frac}) simultaneously. % i.e., the resulting estimates may be different.
However, we have found this to be a very good approximation, computationally feasible, and much more amenable to analysis. % (this is also the method used in \cite{kim1} but without explicit justification).

%% file: sensing.tex
\section{Best and worst sampling sequences} \label{sec:sensing}

The goal of this study is to judiciously schedule a very limited number of sampling times so that the estimation accuracy is least affected.   We first argue intuitively why the commonly used uniform sampling does not perform well when the number of samples allowed is limited.  This motivates us to look for better sampling schemes.   We then present a precise analysis  through the use of Fisher information, in the case of exponential on/off distributions.  In particular, we will show that using this measure, under a certain sparsity condition, uniform sensing is the {\em worst} schedule in terms of its estimation accuracy.   We also derive an upper bound on the Fisher information as well as the sampling sequence achieving this upper bound.  These provide us with useful benchmarks to assess any arbitrary sampling sequence.  We then present a dynamic programming approach to finding the best and worst sampling sequence without the sparsity condition, which provides a further bound on how well any sampling sequence can be expected to perform.

\subsection{An intuitive explanation} \label{subsec:analysis}

Uniform sensing, where samples are taken at constant time intervals,
is a natural, easy-to-implement, and easy-to-analyze scheme.  Specifically, with the on/off durations being exponential the likelihood function has a particularly simple form; there is also a closed-form solution to the maximization of the log likelihood function, see e.g., \cite{kim1}.
However, when sensing is done sparsely, certain problems arise.
%
%On the other hand, there are a number of drawbacks to uniform sensing compared to random sensing, where samples are taken at random time intervals (generated according to certain probability distribution), given the same average sampling rate.
%
One of the first things to note is that since there is no variation across sampling intervals under uniform sensing, the uniform interval in general needs to be upper-bounded in order to catch potential channel state changes that occur over small intervals\footnote{One such upper bound was proposed in \cite{kim1}.}.
This bound cannot be guaranteed under sparse sensing.  If sensing is done randomly, then even if the average sampling interval is large, there can be significant probability for sufficiently small sampling intervals to exist in any realization of the sampling time sequence $\{t_{1}, t_{2}, \cdots, t_{m}\}$.

We show in Figure \ref{fig:quick-comparison} a comparison between uniform sensing and random sensing where the sensing times are randomly placed using a uniform distribution \footnote{Here uniform distribution refers to the sampling times being randomly placed within the window following a uniform distribution, not to be confused with uniform sensing where sampling intervals are a constant.} within a window of $5000$ time units.  The on/off periods are exponentially distributed with parameters $E[T_0]=2$, $E[T_1]=1$ time units, respectively.  The figure shows the estimated value of $E[T_0]$ as a function of the number of samples taken within the window of $5000$.  We see that random sensing outperforms uniform sensing, and significantly so when $m$ is small.
\begin{figure}[htb]
\centering
 \includegraphics[width=0.5\textwidth]{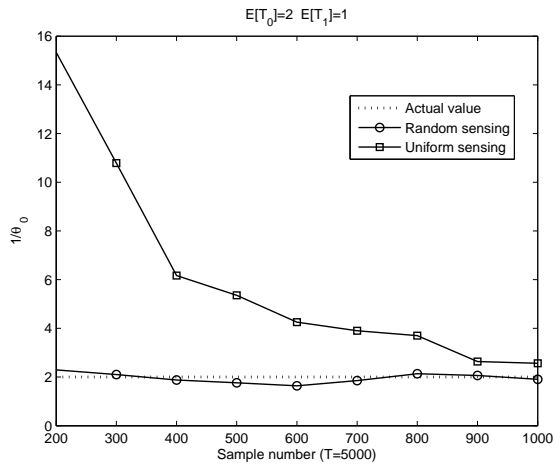}\\
  \centering
  \parbox{9cm}{\caption{Estimation accuracy: uniform sensing vs. random sensing}\label{fig:quick-comparison}}
\end{figure}

The key to the increased accuracy is not so much that we used randomly generated sensing times as is the fact that a randomly generated sequence contains significantly more variability in its sampling intervals.  In this sense a sequence does not have to be randomly generated; as long as it contains sufficient variability, estimation accuracy can be improved.  Random generation is an easy and more systematic way of obtaining such a sequence.

To see why this variability is important when sampling is sparse, consider the transition probabilities $P_{ij}(\Delta t)$, $i, j\in\{0, 1\}$.  As shown in the previous section, these probabilities completely define the likelihood function.  They approach the stationary probabilities as $\Delta t$ increases.  For instance, we have $P_{01}(\Delta t) \rightarrow \frac{E[T_{1}]}{E[T_{1}]+E[T_{0}]} = u$ as $\Delta t\rightarrow \infty$, and so on.  This stationary quantity represents the average fraction of time the channel is busy, which contains little direct information on the average length of a busy period, the parameter we are trying to estimate.  Depending on the mixing time of the underlying renewal process, this convergence can occur rather quickly.  What this means is that if sampling is sparsely done,
%in which case $\Delta t$ is a relatively large constant $\Delta t_o$ under uniform sensing,
then these transition probabilities will become constant-like (i.e., approaching the stationary value).
Loosely speaking, this means that the samples are of a similar quality, each providing little additional information.  %In particular, there is zero probability that we will collect two samples less than $\Delta t_o$ apart.   By contrast, under a random sensing scheme $\Delta t$ is a variable, meaning with positive probability there will be sampling intervals both {\em smaller} and {\em larger} than $\Delta t_o$.  This variability in the interval realizations results in measurements of $P_{01}(\Delta t)$ that contain more information about the parameter to be estimated (we will shortly make this notion more precise).
This also in turn causes the likelihood function to be constant-like, making it difficult for the ML estimator to produce accurate estimates \cite{ML_defination}.
Interestingly, in a similar spirit but for a different problem,  \cite{Tong} studied an information retrieval problem where sensors are queried for data and they may be active or inactive.  It was shown that if the active sensors are sparse, then randomly accessing them outperforms periodic (or uniform) schedules.%The authors in \cite{Tong} considered two MAC schemes for information retrieval in sensor networks. One is deterministic scheduling which collects data from deterministic locations, the other one is random access which collects data from random locations. They showed that the deterministic scheduling provides a better performance when the sensor outage probability, which gives the probability that there is no active sensor within the scheduled resolution interval, is less than a critical threshold. Otherwise, random access outperforms the optimal scheduling. This indicates that if the active sensors are relative {\em sparse}, random access is a better scheduling.

%%%%%%%%%%%%%%%%%%%%%%%%%%%%%%%%%%%%%%%%%%%%%%%%%%%%%%%%%%%%%%%%%%%%%%%%%%%%%%%%%%%%%%%%%%%%%%%%%%%%%%%%%
\subsection{Fisher information and preliminaries} \label{subsec:fisher}

We now analyze this notion of information content more formally via a measure known as the {\em Fisher information} \cite{fisher}.
For the likelihood function given in Eqn (\ref{eqn:likelihood-theta0}),  the Fisher information is defined as:
%It is defined as follows for a given log likelihood function, assumed to be twice differentiable with respect to $\bar{\theta}$:
%\begin{equation}
%\label{fisher}
%I(\bar{\theta})=-E[\frac{\partial^{2}\ln L(\bar{\theta})}{\partial \bar{\theta}^{2}}] ~.
%\end{equation}
\begin{eqnarray}
\label{eqn:fisher}
I(\theta_{0})&=& -E[\frac{\partial^{2}\ln L(\theta_{0})}{\partial\theta^{2}_{0}}]  ~.
\end{eqnarray}
The Fisher information is a measure of the amount of information an observable random variable conveys about an unknown parameter.  This measure of information is particularly useful when comparing two observation methods of random processes (see e.g., \cite{FisherInfo}).  The precision to which we can estimate ${\theta_0}$ is fundamentally limited by the Fisher information of the likelihood function.

%Due to the product form of the likelihood function given in Eqn (\ref{likfun}), the Fisher information can be written as the summation of functions $g(\Delta t_{i};\bar\theta)$, $i=2,3,\cdots, m$, such that $I(\bar\theta) = \sum_{i=2}^{m} g(\Delta t_i;\bar\theta)$, where
%\begin{equation}
%\label{g_(t)}
%g(\Delta t_i;\bar\theta)= -E[\frac{\partial^{2}\ln P_{z_{i-1} z_i}(\Delta t_i;\bar{\theta})}{\partial \bar{\theta}^{2}}] ~.
%\end{equation}
%The function $g()$ will be referred to as the {\em Fisher function} in the remainder of our discussion. Note that $g()$ is a function of both $\Delta t_i$ and $\bar\theta$.  However, we will suppress $\bar\theta$ from the argument in the Fisher function $g()$ and write it simply as $g(\Delta t)$.  This is because our analysis focuses on how this function behaves as we select different $\Delta t$ (the sampling interval) while holding $\bar\theta$ constant.
%Note that the first term in Eqn (\ref{likfun}) does not appear in the above expression.  This is because this first term is only a function of $u$ (see Eqn (\ref{eqn:first-term-u})), which is separately estimated using Eqn (\ref{eqn:channel_utilization}) and not viewed as a function of $\bar\theta$. Therefore the term disappears after the differentiation.

Due to the product form of the likelihood function, we have
%With the likelihood function given in
%Similarly, using the likelihood function given in
%Eqn (\ref{eqn:likelihood-theta0}),  we have:
%the Fisher information can be written as the summation of functions $g(\Delta t_{i}; \theta_0)$, $i=2,3,\cdots, m$, such that
%$I(\theta_0) = \sum_{i=2}^{m} g(\Delta t_i; \theta_0)$, where
%the Fisher information can be written as:
\begin{eqnarray}
%\begin{split}
\label{eqn:fisher-information}
I(\theta_{0}) %&=& -E[\frac{\partial^{2}\ln L(\theta_{0})}{\partial\theta^{2}_{0}}] \nonumber \\
&=& -E\big[\sum_{i=2}^{m}\frac{\partial^{2}\ln[\alpha_{i}+\beta_{i}e^{-\theta_{0}\Delta t_{i}/u}]}{\partial\theta^{2}_{0}}\big] \nonumber \\
&=& \sum_{i=2}^{m}\frac{\Delta
t_{i}^{2}}{u^{2}}E\big[\frac{-\alpha_{i}\beta_{i}e^{-\theta_{0}\Delta
t_{i}/u}}{(\alpha_{i}+\beta_{i}e^{-\theta_{0}\Delta
t_{i}/u})^{2}}\big] ~, %= \sum_{i=2}^{m} g(\Delta t_i) ~,
%\end{split}
\end{eqnarray}
where $\alpha_{i}=u^{z_{i}}(1-u)^{1-z_{i}}$ and $\beta_{i}=(-1)^{z_{i}+z_{i-1}}u^{1-z_{i-1}}(1-u)^{z_{i-1}}$.
Define:
\begin{eqnarray}
g(\Delta t_i; \theta_0) &=& \frac{\Delta
t_{i}^{2}}{u^{2}}E\big[\frac{-\alpha_{i}\beta_{i}e^{-\theta_{0}\Delta
t_{i}/u}}{(\alpha_{i}+\beta_{i}e^{-\theta_{0}\Delta
t_{i}/u})^{2}}\big] ~,
\end{eqnarray}
so that the Fisher information can be simply written as $I(\theta_{0})=\sum_{i=2}^{m}g(\Delta t_{i})$.
The function $g()$ will be referred to as the {\em Fisher function} in our discussion. Note that $g()$ is a function of both $\Delta t_i$ and $\theta_0$.  However, we will suppress $\theta_0$ from the argument and write it simply as $g(\Delta t)$.  This is because our analysis focuses on how this function behaves as we select different $\Delta t$ (the sampling interval) while holding $\theta_0$ constant.
Note that the first term in Eqn (\ref{eqn:likelihood-theta0}) does not appear in the above expression.  This is because this first term is only a function of $u$ (see Eqn (\ref{eqn:first-term-u})), which is separately estimated using Eqn (\ref{eqn:channel_utilization}) and not viewed as a function of $\theta_0$. Therefore the term disappears after the differentiation.

%Again each term within the summation will be referred to as the Fisher function, and the Fisher information can be written as
%$I(\theta_{0})=\sum_{i=2}^{m}g(\Delta t_{i})$.

The expectation on the right-hand side of (\ref{eqn:fisher-information}) can be calculated by considering all four possibilities for the pair ($z_{i-1}, z_i$), i.e., (0, 0), (0, 1), (1, 0), and (1, 1).  Using Eqn (\ref{exp}), we obtain the transition probability of each case to be
$(1-u)P_{00}(\Delta t)$, $(1-u)P_{01}(\Delta t)$, $uP_{10}(\Delta t)$ and $uP_{11}(\Delta t)$, respectively.
We can therefore calculate the Fisher function as follows:
\begin{eqnarray}
\label{eqn:fisher-function}
g(\Delta t)&=& \frac{
\Delta t^{2}}{u^{2}}e^{-\theta_{0}\Delta t/u}\big[\frac{u^{2}(1-u)}{u-ue^{-\theta_{0}\Delta t/u}}+ \frac{u(1-u)^{2}}{(1-u)-(1-u)e^{-\theta_{0}\Delta t/u}} \nonumber \\
&&\quad \quad \quad \quad \quad-\frac{u(1-u)^{2}}{(1-u)+ue^{-\theta_{0}\Delta t/u}} %%\\
-\frac{u^{2}(1-u)}{u+(1-u)e^{-\theta_{0}\Delta t/u}}\big].
\end{eqnarray}

Below we show that under a certain {\em sparsity} condition on the sampling rate, the Fisher function is strictly convex, and that the Fisher information is minimized when uniform sampling is used.  We begin by introducing this sparsity condition.

\begin{condition}(Sparsity condition)\label{cond:sparsity}
Let $\alpha = \max\{2+\sqrt{2}, \ln(\frac{1-u}{u}),
\ln(\frac{u}{1-u})\}$.  This condition requires that $\Delta
t>\alpha u/ \theta_0$.
\end{condition}

Taking $\Delta t$ to be the time between two consecutive sampling points, the above condition states that these two points cannot be too close together with respect to the average off duration ($1/\theta_0$) and the channel utilization $u$.

\begin{lem}\label{lem:convexity}
%Let $\alpha = \max\{2+\sqrt{2}, \ln(\frac{1-u}{u}), \ln(\frac{u}{1-u})\}$.  Then
The Fisher function $g(\Delta t)$ given in Eqn (\ref{eqn:fisher-function}) is strictly convex under Condition \ref{cond:sparsity} (i.e, for $\Delta t>\alpha u/ \theta_0$).
\end{lem}

The proof of this lemma can be found in the Appendix. Using this lemma we next derive tight lower and upper bounds of the Fisher information.

%%%%%%%%%%%%%%%%%%%%%%%%%%%%%%%%%%%%%%%%%%%%%%%%%%%%%%%%%%%%%%%%%%%%%%%%%%%%%%%%%%%%%%%%%%%%%%%%%%%%%%%%%
\subsection{A tight lower bound on the Fisher information}
\begin{lem}\label{lem:minimum}
For any $n\in \mathbb{N}, n\geq 1$, $T\in \mathbb{R},T> (n+1) \alpha u/ \theta_0$, and $\alpha u/ \theta_0 < \Delta t < T-n \alpha u/\theta_0$, %where $\alpha = \max\{2+\sqrt{2}, \ln(\frac{1-u}{u}), \ln(\frac{u}{1-u})\}$,
the function $G(\Delta t)=ng(\frac{T-\Delta t}{n})+g(\Delta t)$ has a minimum of $(n+1)g(\frac{T}{n+1})$ attained at $\Delta t=\frac{T}{n+1}$.
\end{lem}

\begin{proof}
Setting the first derivative of $G$ to zero and solving for $\Delta t$ results in solving the equation $g^{'}(\Delta t) = g^{'}(\frac{T-\Delta t}{n})$.  Since the arguments on both side satisfy Condition 1, by the assumption of the lemma, $g$ is strictly convex according to Lemma 1 and $g^{'}$ is a strictly monotonic function. Therefore there exists a unique solution within the range of $(\alpha u/\theta_0, T-n\alpha u/\theta_0)$ to this equation at $\Delta t=\frac{T}{n+1}$.

Next we calculate the second derivative of $G$ at this point.
Since $G^{''}(\Delta t)=g^{''}(\Delta t)+\frac{1}{n}g^{''}(\frac{T-\Delta t}{n})$, we have $G^{''}(\frac{T}{n+1})=(1+\frac{1}{n})g^{''}(\frac{T}{n+1})$. Since $T> (n+1)\alpha u/\theta_0$, $g$ is convex at this stationary point by Lemma 1.  Hence $G$ is convex at this point and it is thus a global minimum within the range $(\alpha u/\theta_0, T-n\alpha u/\theta_0)$; the minimum value is $(n+1)g(\frac{T}{n+1})$, completing the proof.
\end{proof}

\begin{thm}
\label{thm:min}
Consider a period of time $[0, T]$, in which we wish to schedule $m\geq3$ sampling points, including one at time $0$ and one at time $T$.  Denote the sequence of time spacings between these samples as $\underline{\Delta t} = [\Delta t_2, \Delta t_3, \cdots, \Delta t_m]$, where $\sum_{i=2}^{m} \Delta t_i= T$.  For a given sequence $\underline{\Delta t}$, define the Fisher information $I(\theta_0)$ as in Eqn (\ref{eqn:fisher-information}) and rewrite it as $I(\theta_0; \underline{\Delta t})$ to emphasize its dependence on $\underline{\Delta t}$.  Assuming $T>(m-1)\alpha u/\theta_0$, then we have
\begin{eqnarray*}
%&& \min_{\sum_{i=2}^{m} \Delta t_i = T, \Delta t_i > \alpha u/\theta_0, i=2, \cdots, m}
 \min_{\underline{\Delta t} \in {\cal A}_m}  I(\theta_0; \underline{\Delta t})
= (m-1) g(\frac{T}{m-1}),
\end{eqnarray*}
where ${\cal A}_m = \{\Delta t_i: \sum_{i=2}^{m} \Delta t_i = T, \Delta t_i > \alpha u/\theta_0, i=2, \cdots, m\}$, and
with the minimum achieved at $\Delta t_{i}=\frac{T}{m-1}, i= 2, \cdots, m$.
\end{thm}

\begin{proof}
We prove this by induction on $m$.

\emph{Induction basis:} For $m=3$,
\begin{eqnarray*}
I(\theta_{0}; \underline{\Delta t})=g(\Delta t_2) + g(\Delta t_3).
\end{eqnarray*}
Using Lemma $1$ in the special case of $n=1$ the result follows.

\emph{Induction step:} Suppose the result holds for $3, 4,\dots m$, we want to show it also holds for $m+1$ for $T>m \alpha u/\theta_0$.
Note that in this case $\underline{\Delta t} \in {\cal A}_{m+1}$ implies that $\alpha u/\theta_0 < \Delta t_{m+1} < T-(m-1)\alpha u/\theta_0$, which will be denoted as $\Delta t_{m+1}\in {\cal A}_{m+1}$ below for convenience.  We thus have
\begin{eqnarray}
&& \min_{\underline{\Delta t} \in {\cal A}_{m+1}} \{I(\theta_{0}; \underline{\Delta t})\} \nonumber \\
&=& \min_{\underline{\Delta t} \in {\cal A}_{m+1}} \left\{\sum_{i=2}^{m}g(\Delta t_{i})+g(\Delta t_{m+1}) \right\} \nonumber \\
&=& \min_{\Delta t_{m+1}\in {\cal A}_{m+1}} \Bigg{\{} \min_{\sum \Delta t_i= T-\Delta t_{m+1}}
\Bigg{\{} \sum_{i=2}^{m}g(\Delta t_{i}) \Bigg{\}}+ g(\Delta t_{m+1}) \Bigg{\}}   \nonumber \\
&=& \min_{\Delta t_{m+1}\in {\cal A}_{m+1}} \left\{ (m-1)g(\frac{T-\Delta t_{m+1}}{m-1})+g(\Delta t_{m+1}) \right\} \nonumber \\
&=& m g(\frac{T}{m}) ~, \nonumber
\end{eqnarray}
where the third equality is due to the induction hypothesis and the first term on the RHS is obtained at $\Delta t_{i}=\frac{T-\Delta t_{m+1}}{m-1}$, $i=2, \dots, m$.  The last equality invokes Lemma \ref{lem:minimum} in the special case of $n=m-1$, and is obtained at $\Delta t_{m+1}=\frac{T}{m}$.  Combining these we conclude that the minimum value of Fisher information is $mg(\frac{T}{m})$, when $\Delta t_{i}=\frac{T}{m}, i=2, \dots, m+1$. Thus the case $m+1$ also holds, completing the proof.
\end{proof}

Theorem 1 states that given the total sensing period $T$ and the total number of samples $m$, provided that the sampling is done sparsely (with sufficiently large sampling intervals as defined in Condition \ref{cond:sparsity}), the Fisher information attains its minimum when all sampling intervals have the same value, i.e when using a uniform sensing schedule.  In this sense uniform sensing is the {\em worst} possible sensing scheme; any deviation from it, while keeping the same average sampling interval $T/(m-1)$, can only increase the Fisher information.   As we have seen in Figure \ref{fig:quick-comparison}, this increase in Fisher information becomes more significant when sampling gets sparser, i.e., when $m$ decreases.

%This theorem says that among all sampling sequences with sampling intervals satisfy the sparsity Condition \ref{cond:sparsity}, the worst as measured by the Fisher information is the uniform sampling sequence; any deviation from this strategy improves the Fisher information.

%%%%%%%%%%%%%%%%%%%%%%%%%%%%%%%%%%%%%%%%%%%%%%%%%%%%%%%%%%%%%%%%%%%%%%%%%%%%%%%%%%%%%%%%%%%%%%%%%%%%%%%%%
\subsection{A tight upper bound on the Fisher information}
The derivation of the upper bound follows very similar steps as those for the lower bound.

\begin{lem}\label{lem:maximum}
For any  $T\in \mathbb{R},T>2\alpha u/\theta_0$, and $\alpha u/\theta_0<\Delta t < T- \alpha u/\theta_0$, the function $F(\Delta t)=g(T-\Delta t)+g(\Delta t)$ has a maximum of $g(\alpha u/\theta_0)+g(T-\alpha u/\theta_0)$ attained at $\Delta t=\alpha u/\theta_0$ or  $\Delta t=T-\alpha u/\theta_0$.
\end{lem}

\begin{proof}
Firstly we prove that $F$ is convex under the stated conditions.  We have
\begin{displaymath}
F^{'}(\Delta t)=g^{'}(\Delta t)-g^{'}(T-\Delta t) ~.
\end{displaymath}
Since $g$ is strictly convex under the stated conditions, by Lemma \ref{lem:convexity} $g^{'}$ is monotonic increasing. Thus $F^{'}$ is also monotonic increasing, hence $F$ is convex. It follows that the maximum of $F(\Delta t)$ is attained at one and/or the other extreme point of $\Delta t$.  In either case we have
%
% -- do we need to specify which end point, or that it depends??}
%
%We have proven that $F$ has a minimum value attained at $T/2$, then its maximum value should be attained at one of its bound.
\begin{displaymath}
F(\alpha u/\theta_0)=F(T-\alpha u/\theta_0)=g(\alpha u/\theta_0)+g(T-\alpha u/\theta_0).
\end{displaymath}
\end{proof}

\begin{thm}
\label{thm:max}
Consider a period of time $[0, T]$, in which we wish to schedule $m\geq3$ sampling points, including one at time $0$ and one at time $T$.  Denote the sequence of time spacings between these samples as $\underline{\Delta t} = [\Delta t_2, \Delta t_3, \cdots, \Delta t_m]$, where $\sum_{i=2}^{m} \Delta t_i= T$. Assuming $T>(m-1)\alpha u/\theta_0$, then we have
\begin{eqnarray*}
\max_{\underline{\Delta t} \in {\cal A}_m}  I(\theta_0; \underline{\Delta t})
= (m-2) g(\alpha u/\theta_0)+g(T-(m-2)\alpha u/\theta_0),
\end{eqnarray*}
where ${\cal A}_m = \{\Delta t_i: \sum_{i=2}^{m} \Delta t_i = T, \Delta t_i > \alpha u/\theta_0, i=2, \cdots, m\}$, and
with the maximum achieved at $\Delta t_{i}=\alpha u/\theta_0, i= 2, \cdots, m-1$ and $\Delta t_{m}=T-(m-2)\alpha u/\theta_0$.

\end{thm}

\begin{proof}
We prove this by induction on $m$.

\emph{Induction basis:} For $m=3$,
%\begin{eqnarray*}
$I(\theta_{0}; \underline{\Delta t})=g(\Delta t_2) + g(\Delta t_3)$.
%\end{eqnarray*}
Using Lemma \ref{lem:maximum} the result immediately follows.

\emph{Induction step:} Suppose the result holds for $3, 4,\dots m$, we want to show it also holds for $m+1$ for $T> m\alpha u/\theta_0$.
Again in this case $\underline{\Delta t} \in {\cal A}_{m+1}$ implies that $\alpha u/\theta_0 < \Delta t_{m+1} < T-(m-1)\alpha u/\theta_0$, which will be denoted as $\Delta t_{m+1}\in {\cal A}_{m+1}$ for convenience. We thus have
\begin{eqnarray}
&&\max_{\underline{\Delta t} \in {\cal A}_{m+1}} \{I(\theta_{0}; \underline{\Delta t})\} \nonumber \\
&=& \max_{\underline{\Delta t} \in {\cal A}_{m+1}} \left\{\sum_{i=2}^{m}g(\Delta t_{i})+g(\Delta t_{m+1}) \right\} \nonumber \\
&=& \max_{\Delta t_{m+1}\in {\cal A}_{m+1}} \Bigg{\{}  \max_{\sum \Delta t_i= T-\Delta t_{m+1}}
\left\{ \sum_{i=2}^{m}g(\Delta t_{i}) \right\}+g(\Delta t_{m+1}) \Bigg{\}}   \nonumber \\
&=& \max_{\Delta t_{m+1}\in {\cal A}_{m+1}} \Bigg{\{} (m-2)g(\alpha u/\theta_0)+g(T-\Delta t_{m+1}-(m-2)\alpha u/\theta_0)+g(\Delta t_{m+1}) \Bigg{\}} \nonumber \\
&=& (m-1) g(\alpha u/\theta_0)+g(T-(m-1)\alpha u/\theta_0) ~, \nonumber
\end{eqnarray}
where the third equality is due to the induction hypothesis and the first term on the RHS is obtained at $\Delta t_{i}= \alpha u/\theta_0$, $i=2, \dots, m-1$ and $\Delta t_{m}=T-\Delta t_{m+1}-(m-2)\alpha u/\theta_0$.  The last equality invokes Lemma \ref{lem:maximum}, and is obtained at $\Delta t_{m+1}=T-(m-1)\alpha u/\theta_0$ or $\Delta t_{m+1}=\alpha u/\theta_0$.  Thus the case $m+1$ also holds, completing the proof.
\end{proof}

We see from this theorem that under the sparsity condition, the best sensing sequence is to sample at the smallest interval that the condition would allow, till we use all the $m-2$ samples we have the freedom of placing.  This produces a uniform sequence of sampling times except for the last one.  It can be shown that if we remove the constraint of having a window of $T$, but rather seek to optimally place $m$ points subject to the sparsity condition, then the optimal sequence would be exactly uniform with the interval $\Delta t_i= \alpha u/\theta_0$.
However, since $\theta_0$ is the very thing we are trying to estimate, it would be unreasonable to suggest that this optimal interval is known a priori.  Therefore, this optimal sequence, while exists, is not in general implementable.

%%%%%%%%%%%%%%%%%%%%%%%%%%%%%%%%%%%%%%%%%%%%%%%%%%%%%%%%%%%%%%%%%%%%%%%%%%%%%%%%%%%%%%%%%%%%%%%%%%%%%%%%%%%%%
\subsection{Best and worst sampling schemes without the sparsity condition}

 The preceding upper- and lower-bound achieving sensing sequences were derived under the sparsity Condition 1.
%Unfortunately, similar analysis is not available when we relax this condition.  We also note that if we are to use a random sensing scheme, then depending on how the random sampling intervals are generated, this condition cannot always be guaranteed for all intervals.
Below we show how to obtain the best and worst sensing sequences in a more general setting, without the requirement of Condition 1, via the use of dynamic programming.  %The resulting bounds thus apply to a much broader class of sensing sequences,
While this result is more general compared to those derived under the sparsity condition, structurally they are not as easy to identify and are thus given in a numerical form.
These sequences are also not practically implementable as they also assume the a priori knowledge of the parameters to be estimated.

Denote by $\pi$ a sampling policy given by the time sequence $\{ t_1,  t_1, \cdots,  t_m\}$. Then the optimal sampling policy is given by
\begin{eqnarray}
\pi^*=\arg\max_{\pi \in {\Pi}} I(\theta_0) ~,
%\pi_s^w=\arg\min_{\pi_s} I(\theta_0).
\end{eqnarray}
where the set of admissible policies $\Pi = \{t_i: t_1 = 0, t_m = T, 0< t_2 < \cdots < t_{m-1} < T\}$.

The maximum $I(\theta_0)$ can be recursively solved through the set of dynamic programming equations given below:
\begin{eqnarray}
\label{DP_value}
V(1, t) &=& g(T-t), ~~  \forall ~0\leq t < T;  \nonumber \\
V(k, t) &=& \max_{t < x < T}[g(x-t)+V(k-1, x)], ~~ \forall ~ 0 \leq t < T, ~k=2,3,\cdots, m-1 ~,
%V(1,t)=0, ~~ m-2\leq t \leq T-1; \nonumber \\
%V(k, t)=\max_{\mbox{\tiny $\begin{array}{c}t <  x \\ m-k\leq x \leq T-k+1 \end{array}$}}[g(x-t)+V(k-1, x)], ~~ m-k-1 \leq t \leq T-k,  1\leq k \leq m-1 ~,
\end{eqnarray}
and
\begin{eqnarray}
\max I(\theta_0) &=& \max_{0 < t < T}[g(t)+V(m-1, t)] ~.
\end{eqnarray}
Here the value function $V(k, t)$ denotes the maximum achievable Fisher
information given we last sampled at time $t$, with $k$ points remaining to be
placed between $(t,T]$.

%With the above equations, we can sequentially obtain the best sampling times by going backwards in time.
%Specifically, $t_m^*=T$, $t_1^*=0$ by requirement.
Note that since $t$ is continuous, the pair $(k, t)$ has an uncountable state space.  In computing the DP equation (\ref{DP_value}) we discretize $t$ and $T$ into small steps and require that both be integer multiples of this small quantity.  The resulting DP has a finite state space and can be solved backwards in time in a standard manner.

It is straightforward to see the exact same procedure can be used to find the sampling sequence that {\em minimizes} the Fisher information, thus giving the worst sampling sequence.  It turns out that the worst sampling sequence in this case coincides with the worst sequence derived under the sparsity condition, i.e., it is also the uniform sequence.
%\footnote{By the same argument we can obtain the worst sampling strategy.}.

\subsection{A comparison}

We now compare the different sensing sequences we obtained in this section using an example.
%Two examples of different kind of sampling time sequences
They are illustrated in Figure \ref{fig:sam_point-a}.
In this example the channel
parameters are $E[T_0]=5$ and $E[T_1]=3$ time units, respectively.
The time window is set to be $40$ time units, and the channel can only be sensed $5$ times.
Shown in the figure are the uniform sensing sequence,
the best/worst sensing sequences derived under the sparsity condition, and the best/worst
sequences derived using dynamic programming.
As mentioned earlier, the worst obtained via dynamic programming coincides with the
uniform sampling sequence.  The worst under the sparsity condition also coincides with the uniform sequence, a fact proven in Theorem \ref{thm:min}, as the sparsity condition holds in this case.
In Figure \ref{fig:sam_point-b}, we compared the performance of these sampling strategies, by setting the time window to $5000$ time units.  The estimated value under each strategy is shown as a function of the number of samples taken.  The true value is also shown for comparison.  These are used as benchmarks in the next section in evaluating random sensing schemes.
\begin{figure}[htb]
\centering
\subfigure[Illustration of different sampling sequences]{\label{fig:sam_point-a}\includegraphics[width=0.45\textwidth,height=0.2\textwidth]{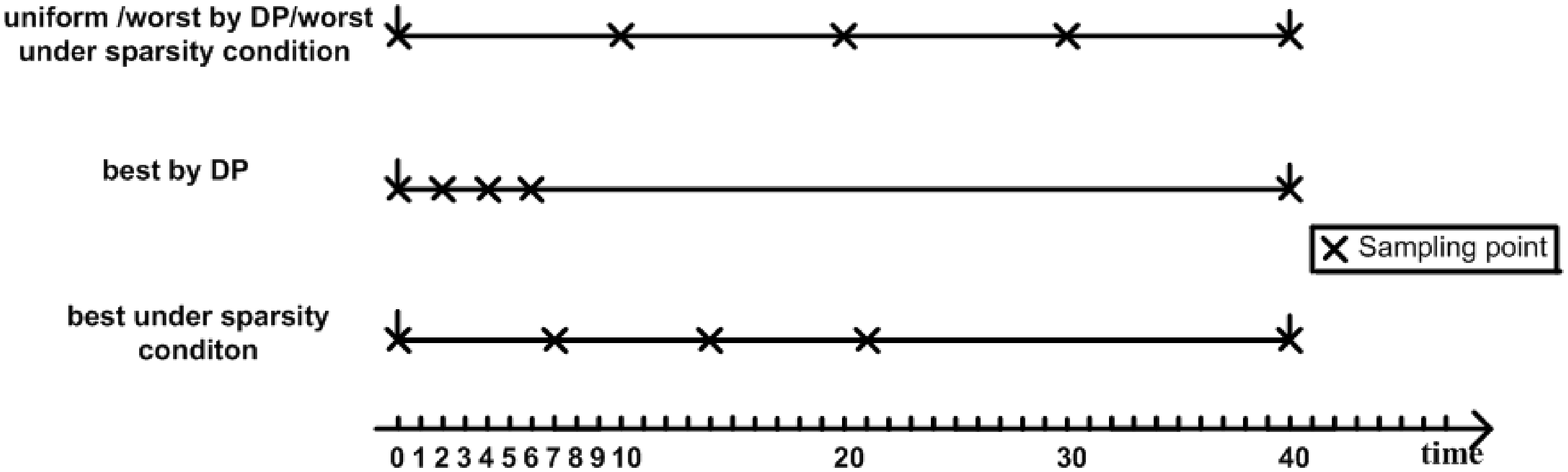}}
\hspace{-0.2in}
\centering
\subfigure[Performance comparison]{\label{fig:sam_point-b}\includegraphics[width=0.5\textwidth]{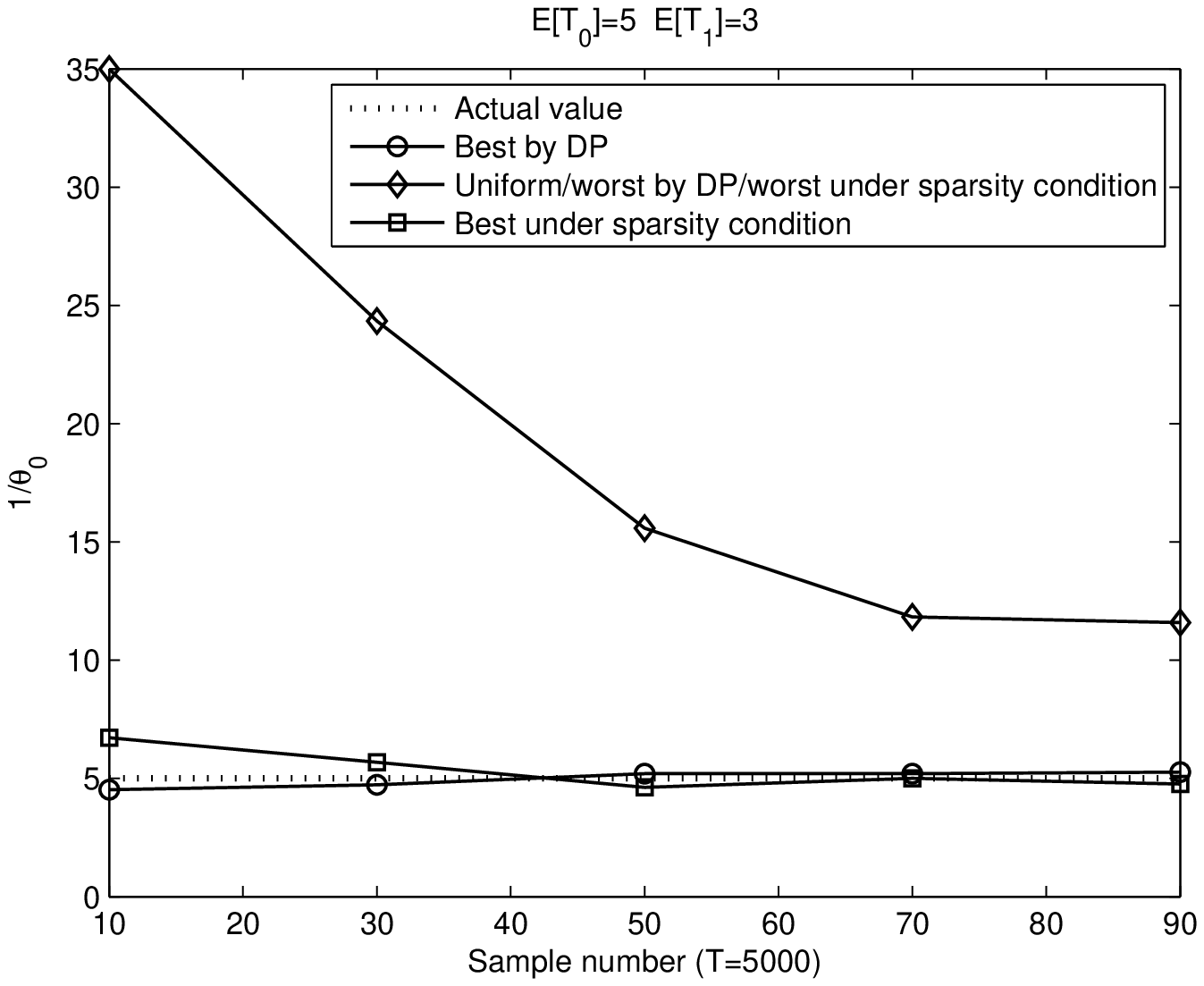}}
\centering
  \parbox{8cm}{\caption{Comparison of different sampling sequence  }\label{fig:sam}}

\end{figure}

%
%In Figure \ref{fig:perf}, we show the performance comparison between these sampling sequences. In the simulation the time window is set to be $5000s$. In Figure \ref{fig:perf-a} the channel
%parameters are $E[T_0]=2s$, $E[T_1]=1s$, and in Figure \ref{fig:perf-b} $E[T_0]=5s$, $E[T_1]=3s$.
%\begin{figure}[htb]
%\subfigure[{$E[T_0]=2s$, $E[T_1]=1s$}]{\label{fig:perf-a}\includegraphics[width=0.5\textwidth]{DP_2_1}}
%\subfigure[{$E[T_0]=5s$, $E[T_1]=3s$}]{\label{fig:perf-b}\includegraphics[width=0.5\textwidth]{DP_5_3}}
%
% \centering \parbox{8cm}{ \caption{ Performance comparison of deterministic sampling}\label{fig:perf}}
%
%\end{figure}

As we can see from Figure \ref{fig:sam_point-a}, the best sensing sequence produced by dynamic programming without the sparsity condition also appears to be uniform except for the last sample, as is the case with the best sequence under the sparsity condition\footnote{Note however that this conclusion is drawn empirically from a large amount of numerical experiment in the case of not requiring sparsity.  By contrast, under the sparsity condition the conclusion is drawn analytically in Theorem \ref{thm:max}.}.  The difference is that the former uses a smaller interval value that violates the sparsity condition.  As mentioned earlier, if we were to remove the requirement that one sample be placed at time $T$, then the optimal sequence of $m$ would appear to be uniform (again, this conclusion is drawn empirically in the case of no sparsity requirement, and precisely and analytically in the case of sparsity), with the optimal interval being the value that maximizes (\ref{eqn:fisher-function}).  Interestingly, the worst sequence is also uniform with or without the sparsity condition.

What this result suggests is that in the ideal case if we have a priori knowledge of the channel parameters, to maximize the Fisher information the best thing to do is indeed to sense uniformly.
The difficulty of course is that without this knowledge we have no way of deciding what the optimal interval should be, and uniform sensing would be a bad decision as it could turn out to be the worst with an unfortunate choice of the sampling interval.

In such cases, the robust thing to do is simply to sense randomly, so that with some probability we will have sampling intervals close to the actual optimum.  This is investigated in the next section.

%% file: random.tex
\section{Random sensing} \label{sec:diff_random}

%Note that for the above results to hold we need the condition $\Delta t_i > \alpha u/\theta_0$.  This is because the analysis involves comparing specific sampling sequences, i.e., the average performance (over channel on/off realizations) of an exact sampling sequence. In reality, we are not able to achieve the upper bound of Fisher information since we have no prior information of the channel parameters. However the bound provides us a benchmark to test a sampling scheme. Therefore without any prior information about the channel, random sampling scheme is  a better choice. In this section we compare different random sampling schemes.

Under a random sensing scheme, the sampling intervals $\Delta t_i$ are generated according to some distribution $f(\Delta t)$ (this may be done independently or jointly).  Below we first analyze how the resulting Fisher information is affected, and then use a family of distributions generated by the circular $\beta$ ensemble to examine the performance of different distributions.

\subsection{Effect on the Fisher information}\label{subsec:analytical comparison}
%Assuming that the sampling intervals $\Delta t_i$'s are generated independently according to some pdf $f(\Delta t)$ we now
We begin by examining the expectation of the Fisher function, averaged over randomly generated sampling intervals, calculated as follows:
\begin{eqnarray}
\label{fisherexp}
E[g(\Delta t)]&=& \int_{0}^{\infty}g(\Delta t)f(\Delta t)d\Delta t  \\ \nonumber
&=& \int_{0}^{\infty}[g(\mu_o)+g'(\mu_o)(\Delta t-\mu_o)\\ \nonumber
&&+\cdots+\frac{g^{(n)}(\mu_o)(\Delta t-\mu_o)^{n}}{n!}+\cdots]f(\Delta t)d\Delta t \\ \nonumber
&=&
g(\mu_o)+g'(\mu_o)\mu_{1}+\cdots+\frac{g^{(n)}(\mu_o)\mu_{n}}{n!}+\cdots
\end{eqnarray}
where the Taylor expansion is around the expected sampling interval $\mu_o = E[\Delta t]$, or $T/(m-1)$ for given window $T$ and $m$ number of samples taken, and   %where Taylor expansion is used and
$\mu_{n}= \int_{0}^{\infty}(\Delta t-\mu_o)^{n}f(\Delta t)d\Delta t$ is
the $n$th order central moment of $\Delta t$.

In order to have a fair comparison %(and to ensure the same sparsity condition),
we will assume $T$ and $m$ are fixed, thus fixing the average sampling interval $\mu_o$ under different sampling schemes.  Also note that the value $g^{(n)}(\mu_o)$ is completely determined by the channel statistics and not the sampling sequence.
%In our case the precondition of sampling is that under same sampling rate, that is the means of the sampling interval are equal for different sampling schemes.
Consequently the expected value of the Fisher function is affected by the selection of a sampling scheme only through the higher order central moments of the distribution $f()$.
Note that the expectation of the Fisher function under uniform sampling with constant sampling interval $\mu_o$ is simply $g(\mu_o)$ (i.e., only the first term on the right hand side remains).  Therefore any random scheme would improve upon this if it results in a positive sum over the higher order terms.
While the above equation does not immediately lead to an optimal selection of a random scheme, it is possible to seek one from a family of distribution functions through optimization over common parameters.

Before we proceed with this in the next subsection, we compare the normal, uniform and exponential random sampling schemes using the above analysis. In Table \ref{moments} we  list the higher order central moments of normal, uniform and exponential distributions \footnote{For normal distribution the probability distribution function is cut off at zero and then renormalized.}.  It can be easily concluded that among these three choices the Fisher function has the largest expectation under the exponential distribution.
\begin{table}[h]
\caption{Higher central moments}
\renewcommand{\arraystretch}{1.5}%?????
\centering
\begin{tabular}{|c|c|c|c|}
\hline
& Normal & Uniform & Exponential \\
\hline n is even &
$\frac{n!\sigma^{n}}{(\frac{n}{2})!2^{\frac{n}{2}}}$&
$\frac{\mu_o^{n}}{n+1}$ &\\
\cline{1-3}
n is odd & $0$ & $0$&  \raisebox{1.5ex}[0pt] {$\mu_o^{n}\sum_{k=0}^{n}\frac{(-1)^{k}n!}{k!}$}\\
\hline
\end{tabular}
\label{moments}
\end{table}

We further compare their performance in Fig. \ref{r_com} as we increase the number of samples $m$ over a window of $T=5000$ time units.   Our simulation is done in Matlab and uses a discrete time model; all time quantities are in the same time units.  %The on/off periods are rounded to the closest integers \footnote{We note that the choice of the time unit is rather arbitrary and inconsequential for our purpose.}.
The maximum number of samples is $5000$; this is because the on/off periods are integers, so there is no reason to sample faster than once per unit of time.
The sampling intervals under the uniform sensing are $\lfloor T/(m-1) \rfloor$.  %Under the random sensing scheme, we fix the average sampling interval to be $\lfloor T/(m-1) \rfloor$ and sequentially generate random sampling intervals as follows.
The sampling times under random schemes are generated as follows.
We fix the window $T$ and take $m$ to be the average number of samples\footnote{The reason $m$ is only an average and not an exact requirement is because we cannot guarantee to have exactly $m$ samples within a window of $T$ if we generate sampling intervals randomly according to a given pdf.  By allowing $m$ to be an average we can simply require the pdf to have a mean of $T/(m-1)$.}.  We place the first and the last sampling times at time $0$ and $T$, respectively.  We then sequentially generate $\Delta t_2, \Delta t_2, \cdots$ according to the given pdf $f()$ with parameters normalized such that it has a mean (sampling interval) of $T/(m-1)$.  For each $\Delta t_i$ we generate we place a sampling point at time $\sum_{k=2}^{i} \Delta t_k$.  This process stops when this quantity exceeds $T$.  Note that under this procedure the last sampling interval will not be exactly according to $f()$ since we have placed a sampling point at time $T$.  However, this approximation seems unavoidable.  Alternatively we can allow $T$ to be different from one trial to another while maintaining the same average. As long as $T$ is sufficiently large this procedure does not affect the accuracy or the fairness of the comparison.
For each value of $m$, the result shown on the figure is the average of 100 randomly generated sensing schedules.
We see that exponential random sampling outperforms the other two; this is consistent with our earlier analysis on the Fisher information.

\begin{figure}[htb]
  % Requires \usepackage{graphicx}
  \centering
  \includegraphics[width=0.5\textwidth]{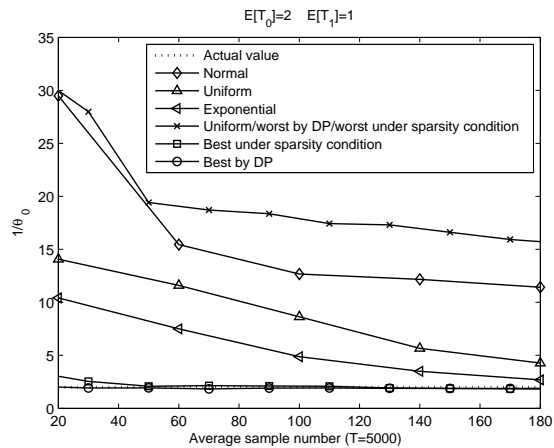}
    \centering
  \parbox{9cm}{\caption{Performance comparison of random sensing : Normal vs. Uniform vs. Exponential}
 \label{r_com}}
\end{figure}

%
%%%%%%%%%%%%%%%%%%%%%%%%%%%%%%%%%%%%%%%%%%%%%%%%%%%%%%%%%%%%%%%%%%%%%%%%%%%%%%%%%%%%%%%%
\subsection{Circular $\beta$  ensemble}
%We have compared the estimation performance using intervals drawn from exponential, uniform and normal distributions. Furthermore we want to learn  about other distributions except for the above three familiar ones.

We now use the circular $\beta$ ensemble \cite{Circular} to study a family of distributions. The advantage of using this ensemble is that with a single tunable parameter we can approximate a wide range of different distributions while keeping the same average sampling rate.

%{\bf (I took this from wikipedia (we need something along these lines, i.e., a better definition of this ensemble): In the theory of random matrices, the circular ensembles are measures on spaces of unitary matrices T as modifications of the Gaussian matrix ensembles. The three main examples are the orthogonal circular ensemble on symmetric unitary matrices, the unitary circular ensemble on unitary matrices, and the symplectic circular ensemble on self dual unitary quaternionic matrices. The unitary circular ensemble is essentially Haar measure on the unitary matrices. However the orthogonal and symplectic circular ensembles are not given by Haar measures, and are not supported on the orthogonal or symplectic groups.}
%
The circular $\beta$ ensemble may be viewed as given by $n$ eigenvalues, denoted as $\lambda_j=e^{i\theta_j}$, $j=1, \cdots, n$.  These eigenvalues have a joint probability density function proportional to the following:
\begin{equation}
\label{Jacobi}
\prod_{1\leq k<l\leq n}|e^{i\theta_k}-e^{i\theta_l}|^\beta,\ \   -\pi<\theta_j\leq\pi, \ \ j,k,l=1,\cdots,n,
\end{equation}
where $\beta>0$ is a model parameter.
%{\bf (Are we saying that these $n$ eigenvalues are the only thing we are interested in?) }
%With each eigenvalue written $\lambda_j=e^{i\theta_j}$, the circular ensemble on $\mathbb{R}^n$, parameterized by $\beta>0$, is the probability distribution whose eigenvalue p.d.f is proportional to
%\begin{equation}
%\label{Jacobi}
%\prod_{1\leq k<l\leq n}|e^{i\theta_k}-e^{i\theta_l}|^\beta,\ \   -\pi<\theta_j\leq\pi, \ \ j,k,l=0,1,\cdots,n,
%\end{equation}
%where $\beta>0$ is a model parameter.
%
In the special cases $\beta=1,2$ and 4, this ensemble describes the joint probability density of the eigenvalues of random orthogonal, unitary and sympletic matrices, respectively \cite{Circular}.

We use the set of eigenvalues generated from the above joint pdf to determine the placement of sample points in the interval $[0, T]$ in the following manner.  In \cite{Jacobi} a procedure is introduced to generate a set of values $\theta_j$, $j=1, 2, \cdots, n$ that follow the joint pdf given by (\ref{Jacobi}).   Setting $n=m$, these $n$ eigenvalues are then placed along a unit circle (each at the position given by $\theta_j$), which are subsequently mapped onto the line segment $[0, 1]$.
%\cite{Jacobi} introduce a matrix $\beta$ Jacobi model which is a distribution on  structured orthogonal matrix, parameterized by $\beta$. Its CS deposition has entries from the ensemble with the same parameter. We use the sampling method introduced in \cite{Jacobi} and calculate the CS values.
%
%We then place the $n$ eigenvalues along a unit circle (each at the position given by $\theta_j$) (we don't mean unit circle with unit radius? we mean a circle with unit perimeter??) and break the circle at $\theta=0$.  This gives us $n$ points randomly placed over a line segment $[0, 2\pi)$.
%
Scaling this segment to $[0, T]$ gives us the $m$ sampling times.
%For $m$ sampling points we will set $n=m-1$???
%
%We can obtain $n$ points randomly placed in [0,1]. Scale it to match  the  time window [0,T], we can obtain the needed sampling points.
The intervals between these points now follow a certain joint distribution.  As $\beta$ varies  we can obtain a family of distributions indexed by $\beta$.  Below we will refer to this method of generating sample points/intervals as using the circular $\beta$ ensemble.
Note that by this procedure we cannot guarantee to have a sample taken at times $0$ and $T$, respectively.  However, since the window size $T$ and the number of samples $m$ are used, we maintained the same average sampling rate.

%
%{\bf The ensemble can be explained as follows: place $n$ points randomly on a circle with unitary perimeter, then the random positions follow the above distribution (Not sure about this explanation.  random placing how? and what does ``position'' mean? euclidian? or is it that we can use the above distribution to place the points?  also how do these point translate to our sampling times?)}. As $\beta$ varies  we can obtain a family of random variables from different distributions. Recall that we want study how to place $m$ points in within a period of $T$ so as to get good estimation performance.  {\bf So the circular ensemble introduced above suits our situation fine, what needed to do is just scaling.}

%In the simulation we use the Matlab code for sampling the $n\times 1$ $\beta$ ensemble introduced in \cite{Jacobi}.
In Fig. \ref{beta_distribution} we give the pdfs of intervals generated by the circular ensemble with different $\beta$.   For each value of $\beta$, We use the generating method in \cite{Jacobi} to obtain $200$ random variables in [0, 1], then scale them to be in [0, 5000]. The  successive intervals between neighboring points are collected with the their pdf shown in the figure. We can see that as $\beta$ approaches $0^{+}$ the pdf becomes exponential-like and as $\beta$ approaches $+\infty$, the pdf becomes deterministic; these are well known facts about circular ensembles.

\begin{figure}[htb]
  % Requires \usepackage{graphicx}
  \centering
  \includegraphics[width=0.5\textwidth]{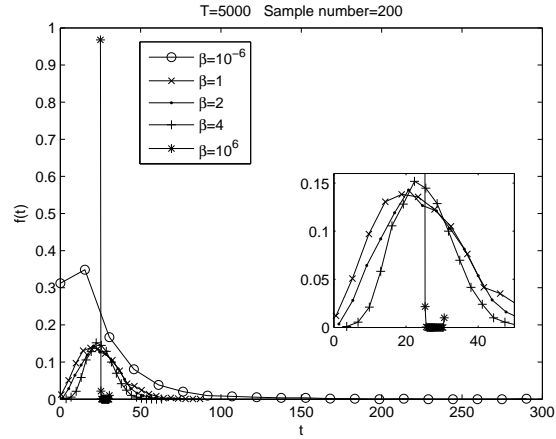}\\
    \centering
  \parbox{9cm}{ \caption{Probability distribution function of intervals generated by the circular $\beta$  ensemble
  }\label{beta_distribution}}
\end{figure}

\subsection{A comparison between different random sensing schemes}

In Fig. \ref{beta_fisher} we show the Fisher information with sampling intervals generated by the circular $\beta$ ensemble.  %The channel on/off durations are exponentially distributed.
The corresponding estimation performance comparison is given in Fig. \ref{com_final}.
The performance of the best and worst sequences with and without the sparsity condition  are also shown for comparison. Note that when $\beta=10^6$, the sampling sequence coincide with the worst obtained via dynamic programming, the worst under sparsity condition and uniform sensing, therefore their performances are the same.
%We also give estimation performance using the optimal and worst sampling strategy by dynamic programming and conditioned ones according to Theorem \ref{thm:min} and \ref{thm:max}.

\begin{figure}[htb]
  % Requires \usepackage{graphicx}
  \centering
  \includegraphics[width=0.5\textwidth]{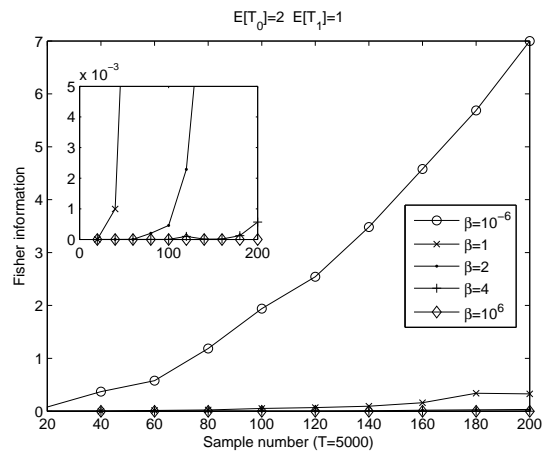}\\
     \centering
  \parbox{9cm}{\caption{Comparison of Fisher information with intervals generated by the circular $\beta$ ensemble for an exponential channel
  }\label{beta_fisher}}
\end{figure}

\begin{figure}[htb]
  % Requires \usepackage{graphicx}
  \centering
  \includegraphics[width=0.5\textwidth]{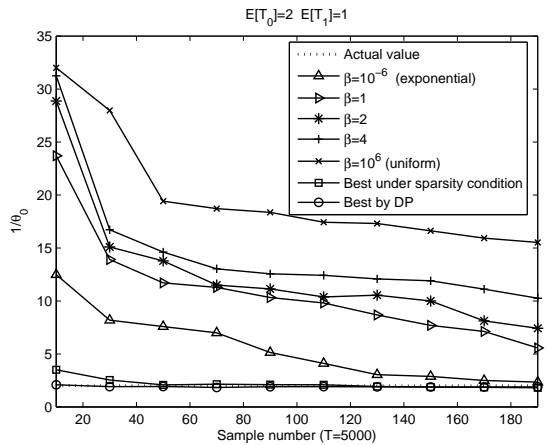}\\
    \centering
  \parbox{9cm}{ \caption{Performance comparison of random sensing with intervals generated by the circular $\beta$ ensemble for an exponential channel
  }\label{com_final}}
\end{figure}

We see again that exponentially generated sampling intervals performs the best.  This may be due to the fact that the on/off durations are also exponentially distributed, thereby creating a good ``match'' between the fisher function $g()$ and the pdf $f()$ that results in a larger value of the expected Fisher function value (see Eqn. (\ref{fisherexp})).
%We end this subsection by illustrating the Fisher function value of exponentially distributed on/off durations with $E[T_{1}]=1$ and $[T_{0}]=2$ time units, respectively, in Fig. \ref{g(t)}.
%Recall that all previous quantitative analysis is based on the assumption that the channel has exponentially distributed on/off durations.
%Our empirical evidence suggests that exponentially generated sampling intervals perform the best.  On possible reason is that the exponential pdf is a good match for the Fisher function $g()$.  Recall that the expectation of the Fisher function is a quadrature of $g$ and pdf $f$ as given in (\ref{fisherexp}).  One might argue that a random sensing scheme with a pdf matching the shape of $g$ will achieve a larger value.  We show three different pdf's in Fig. \ref{g(t)} for comparison.
%This may intuitively explain why in this case exponentially distributed sampling intervals perform the best among all random schemes.
%
%\begin{figure}[htb]
%  % Requires \usepackage{graphicx}
%   \centering
%  \includegraphics[width=0.5\textwidth]{g_t_2_1}\\
%     \centering
%  \parbox{9cm}{\caption{Fisher function of exponentially distributed channel model%: random sensing vs. uniform sensing
%  }\label{g(t)}}
%\end{figure}

%%%%%%%%%%%%%%%%%%%%%%%%%%
\subsection{Discussions on other channel models}

So far all our analysis and results are based on the exponential channel model.  The problem quickly becomes intractable if we move away from this model, though the basic insight should hold. % i.e., with no a priori knowledge on the channel parameters, random sensing would be more robust.
We now examine a channel model with on/off durations following the gamma distribution.  The pdf of the on/off durations are expressed as
\begin{equation}
\label{gamma}
\left\{ \begin{array}{ll}
f_{1}(t)=t^{k_{1}-1}\frac{e^{-t/\lambda_{1}}}{\lambda_{1}^{k_{1}}\Gamma(k_{1})} \\
f_{0}(t)=t^{k_{0}-1}\frac{e^{-t/\lambda_{0}}}{\lambda_{0}^{k_{0}}\Gamma(k_{0})} ~.
\end{array}
\right.
\end{equation}
They are each parameterized by a shape parameter $k$ and a scale parameter $\lambda$, both of which are positive. In this case, the Laplace transforms of $f_{0}(t)$ and $f_{1}(t)$ are $(1+\lambda_{0} s)^{-k_{0}}$ and $(1+\lambda_{1} s)^{-k_{1}}$, respectively, and the expectation of the on/off periods are $E[T_{1}]=k_{1} \lambda_{1}$ and $E[T_{0}]=k_{0}\lambda_{0}$.
In the following simulation both $k_{1}$ and $k_{0}$ are set to $2$, % for computation complexity consideration,
with a simulated time of $5000$ time units. The channel parameters are set to be $E[T_{1}]=10$ and $[T_{0}]=20$ time units. The sampling intervals are randomly generated by the circular $\beta$ ensemble. We see that random sensing again outperforms uniform sensing using such a channel model.
\begin{figure}[htb]
  % Requires \usepackage{graphicx}
  \centering
  \includegraphics[width=0.5\textwidth]{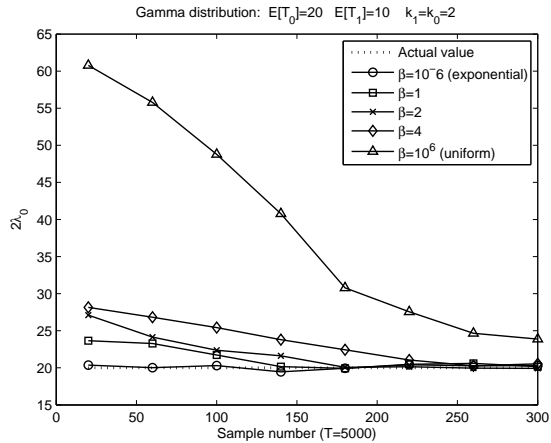}
   \centering
  \parbox{9cm}{\caption{Performance comparison of random sampling with intervals generated by the circular $\beta$ ensemble for gamma channel model}\label{gamma_channel} }
\end{figure}

It should be noted that since the gamma distribution is the conjugate prior for the exponential distribution, the latter being a special case of the former, this result is not surprising.  Unfortunately, obtaining similar result for other channel distributions becomes computationally prohibitive.   The complexity is due to two reasons.  Firstly, for most distributions the Laplace transform is complex, resulting in the complexity in obtaining the corresponding time domain expressions.  Secondly, with the exception of the exponential distribution, without the memoryless property the likelihood function also becomes intractable.

%% file: variation.tex
\section{Adaptive Random Sensing for Parameter Tracking} \label{sec:tracking}

Using insights we have obtained on uniform sensing and random sensing, we now present a method of estimating and tracking a time-varying parameter.  This is a moving window based estimation scheme, where the overall sensing duration $T$ is divided into windows of lengths $T_w$.  In each window samples are taken and an estimate produced at the end of that window.  This estimate is then used to determine the optimal number of samples to be taken in the next window.
This method will be referred to as the {\em adaptive random sensing} scheme. The adaptive nature of the scheme comes from adjusting the number of samples taken in each window based on past estimates.  %This is true whether the sampling

Specifically, at the end of the $i$-th window of $T_w$, we obtain the ML estimate $\hat{\theta}_{0}^{(i)}$ and $\hat{u}^{(i)}$ based on samples collected during that window.
Now {\em assuming} that we will use uniform sensing in the $(i+1)$th window with a sampling interval $\Delta t_p$, and assuming that $\hat{\theta}_{0}^{(i)}$ and $\hat{u}^{(i)}$ are the true parameter values in the $(i+1)$th window, we can obtain the {\em expectation} of the next estimate, denoted as $\tilde{\theta}_{0}^{(i+1)}$, as a function of $(T_w, \Delta t_p, \hat{u}^{(i)}, \hat{\theta}_{0}^{(i)})$.  The optimal sampling interval $\Delta t_p^{(i+1)}$ for the $(i+1)$th window is then calculated as follows:
\begin{equation}\label{equ9}
\Delta t_{p}^{(i+1)}=\arg\min_{\Delta t_{p}}
\bigg|{|\tilde{\theta}_{0}^{(i+1)} - \hat{\theta}_{0}^{(i)}|}-\varepsilon\bigg| ~,
\end{equation}
where $\varepsilon$ is an error factor introduced to lower bound the minimizing interval $\Delta t_p^{(i+1)}$. Without this factor the interval will end up being very small, i.e., requiring a large number of samples for the next window.  The intuition behind the above formula is that assuming the channel parameters are relatively slow varying in time, the estimate from the previous window $\hat{\theta}_{0}^{(i)}$ may be viewed as true.  So for the next window we would like to find the sampling interval that allows us to get as close as possible to this value subject to an error.

Note that the above calculation relies on the availability of $\tilde{\theta}_{0}^{(i+1)}$, a quantity obtained assuming uniform sampling will be used in the next window.  In the actual execution of the algorithm, we simply use this to obtain $\Delta t_{p}^{(i+1)}$ as shown above.  This gives us the desired number of samples to be taken in the next window: $M^{(i+1)} =\lceil { T_w/\Delta t_{p}^{(i+1)}}\rceil$.  Following this, random sensing is used to generate $M^{(i+1)}$ random sampling times within the next window.  An estimate is then made and this process repeats.

It remains to show how $\tilde{\theta}_{0}^{(i+1)}$ is obtained.  As mentioned earlier, when the on/off periods are exponentially distributed there is a simple closed-form solution to the ML estimator.  This was calculated in \cite{kim1} and we will use that result directly below.
Specifically, with $M=\lceil {T_w/\Delta t_p}\rceil$ samples uniformly taken, the estimate of channel utilization $u$ is given by
$\hat{u}=\frac{1}{M}\sum_{i=1}^{M}z_{i}$. The estimate of
$\theta_{0}$ is given by
\begin{equation}\label{equ6}
\hat{\theta}_{0}=-\frac{u}{\Delta t_{p}}\ln[\frac{-B+\sqrt{B^2-4AC}}{2A}],
\end{equation}
where
\begin{equation}\label{equ7}
\left\{ \begin{array}{l}
A=(u-u^2)(M-1)\\B=-2A+(M-1)-(1-u)n_{0}-u n_{3} \\C=A-u n_{0}-(1-u)n_{3}\end{array}
\right. ~.
\end{equation}
Here $n_{0}/n_{1}/n_{2}/n_{3}$ denotes the number of
$(0\rightarrow0)/(0\rightarrow1)/(1\rightarrow0)/(1\rightarrow1)$
transitions out of the total $(M-1)$ transitions. Their respective expectations are given by
\begin{equation}
\label{equ8} \begin{array}{ll} E[n_{0}]=M(1-u)P_{00}(\Delta t_{p};\theta_0),&E[n_{2}]=M u P_{10}(\Delta t_{p}; \theta_0),\\
E[n_{1}]=M(1-u)P_{01}(\Delta t_{p};\theta_0),&E[n_{3}]=M u P_{11}(\Delta t_{p};\theta_0).\\
\end{array}
\end{equation}
Taking these quantities into (\ref{equ7}) and (\ref{equ6}), we obtain the expectation of $\hat{\theta}_{0}$, $\tilde{\theta}_{0}$, which is a function of $(T_w, \Delta t_{p}, u, \theta_{0})$.  Replacing $u$ with $\hat{u}^{(i)}$, $\theta_{0}$ with $\hat{\theta}_{0}^{(i)}$, and $\tilde{\theta}_{0}$ with $\tilde{\theta}_{0}^{(i+1)}$ we obtain the desired result.

Figure \ref{tr} shows the tracking performance of the adaptive
random sensing algorithm, where within each moving window the sampling times are randomly place following a uniform distribution. 
In the simulation the size of
the time window is set to be $3500$ time units and the error factor $\varepsilon$ is
set at $1$. In Figure \ref{tr_off_a} the channel
parameter $E[T_{0}]$ varies as a step function: starting from $6$ time units, it is increased by $5$ 
every $30000$ time units, while $E[T_{1} ]$ is set to $E[T_{0} ]/2$. In Figure \ref{tr_off_b} the channel
parameter changes more smoothly as shown.  %changed once every $1000$s.  
The dashed line represents the actual channel parameter.  For comparison purpose we also include the results
from an adaptive uniform sensing algorithm.  These are obtained by
following the exact same adaptive procedure outlined above, with the
only difference that in the $i$-th window
uniform sensing is used, instead of random sensing, with a constant
sampling interval of $\Delta t_p^{(i)}$.
We see that the estimation under adaptive random sensing (RS) can
closely track the time-varying channel, and clearly outperforms adaptive uniform
sensing (US) at short on/off periods.  %{\bf(It does seem that US catches up with RS
%toward the end?? I'm also thinking whether we could use this adaptive scheme to find the
%optimal sampling interval when the parameter value does NOT change over time.  That is, we sampling window by window. in each window we uniformly sense, and start with some arbitrary interval.  we then use the exact same scheme to adjust the sampling interval over time.  What would happen then?  Do we converge to the best interval?)}

The number of samples taken in each window (or estimation cycle) following this adaptive scheme is given in Figure \ref{track}. It shows as the on/off periods increase, the sampling rate is automatically decreased as an outcome of the tracking.

\begin{figure}[htb]
\subfigure[]{\label{tr_off_a}\includegraphics[width=0.45\textwidth]{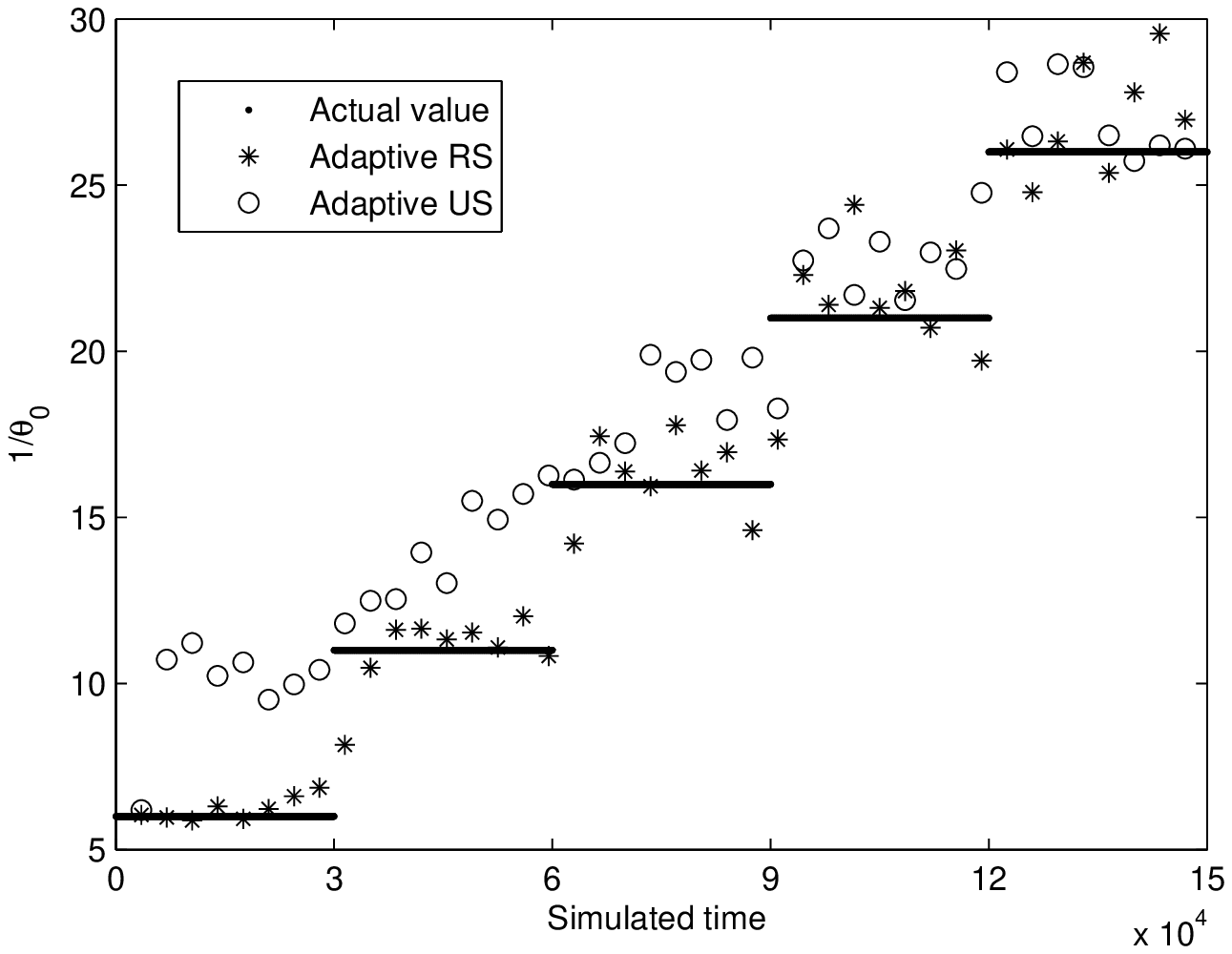}}
\subfigure[]{\label{tr_off_b}\includegraphics[width=0.45\textwidth]{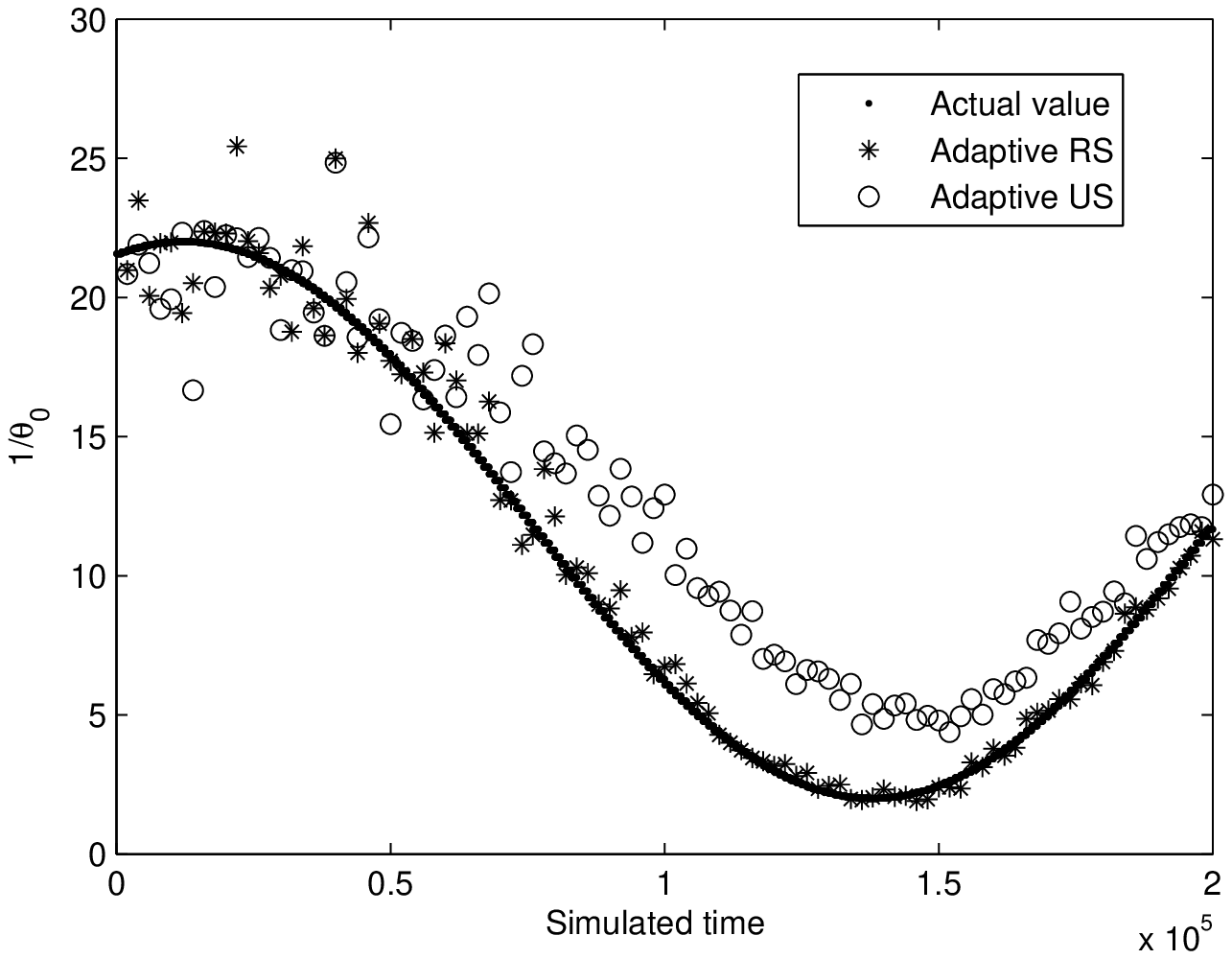}}
  \centering
  \parbox{9cm}{ \caption{Estimation performance of time-varying channel}\label{tr}}
\end{figure}

\begin{figure}[htb]

\centering
% Requires \usepackage{graphicx}
\subfigure[]{ \label{tr_num_a} \includegraphics[width=0.45\textwidth]{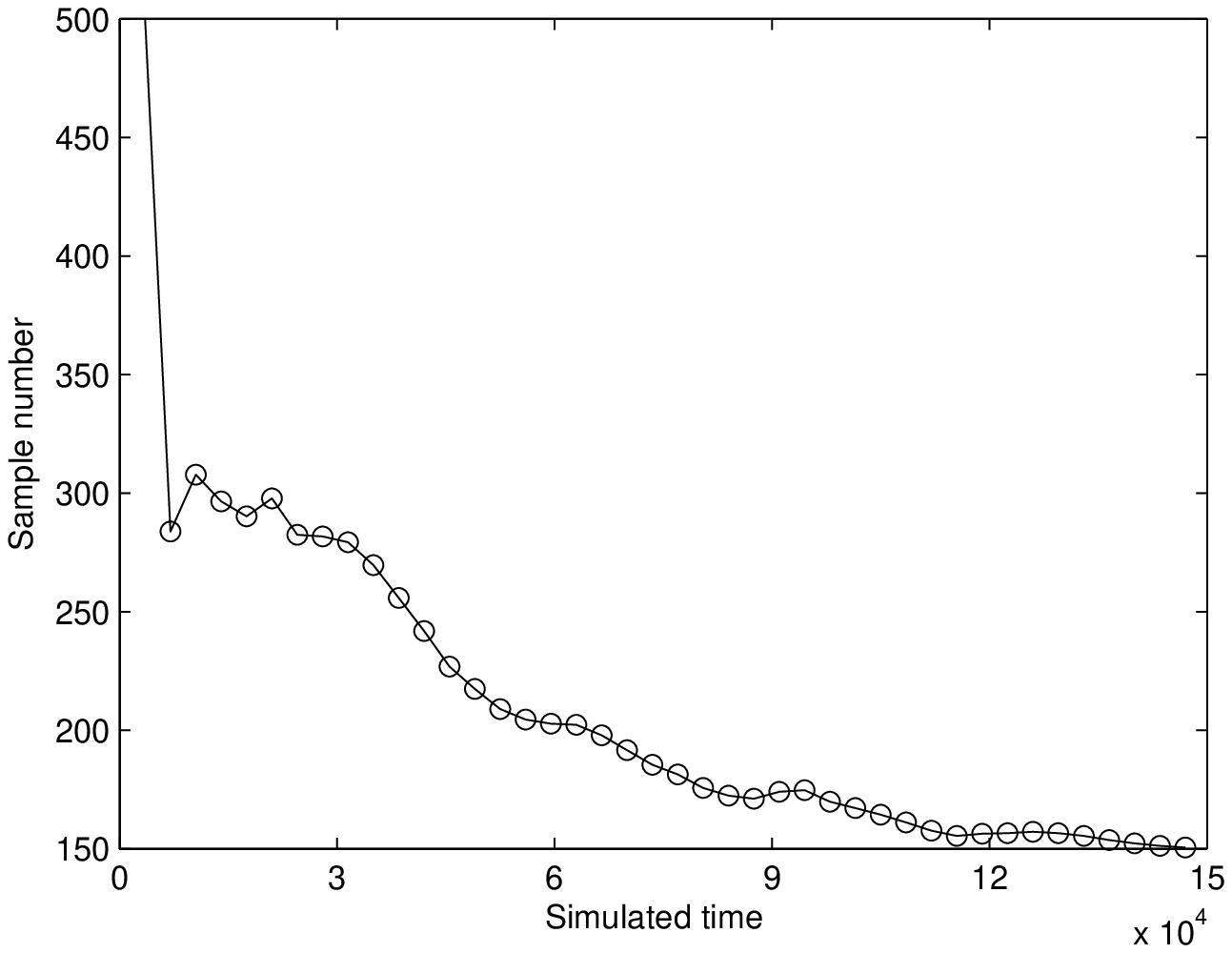}}
\subfigure[]{ \label{tr_num_b} \includegraphics[width=0.45\textwidth]{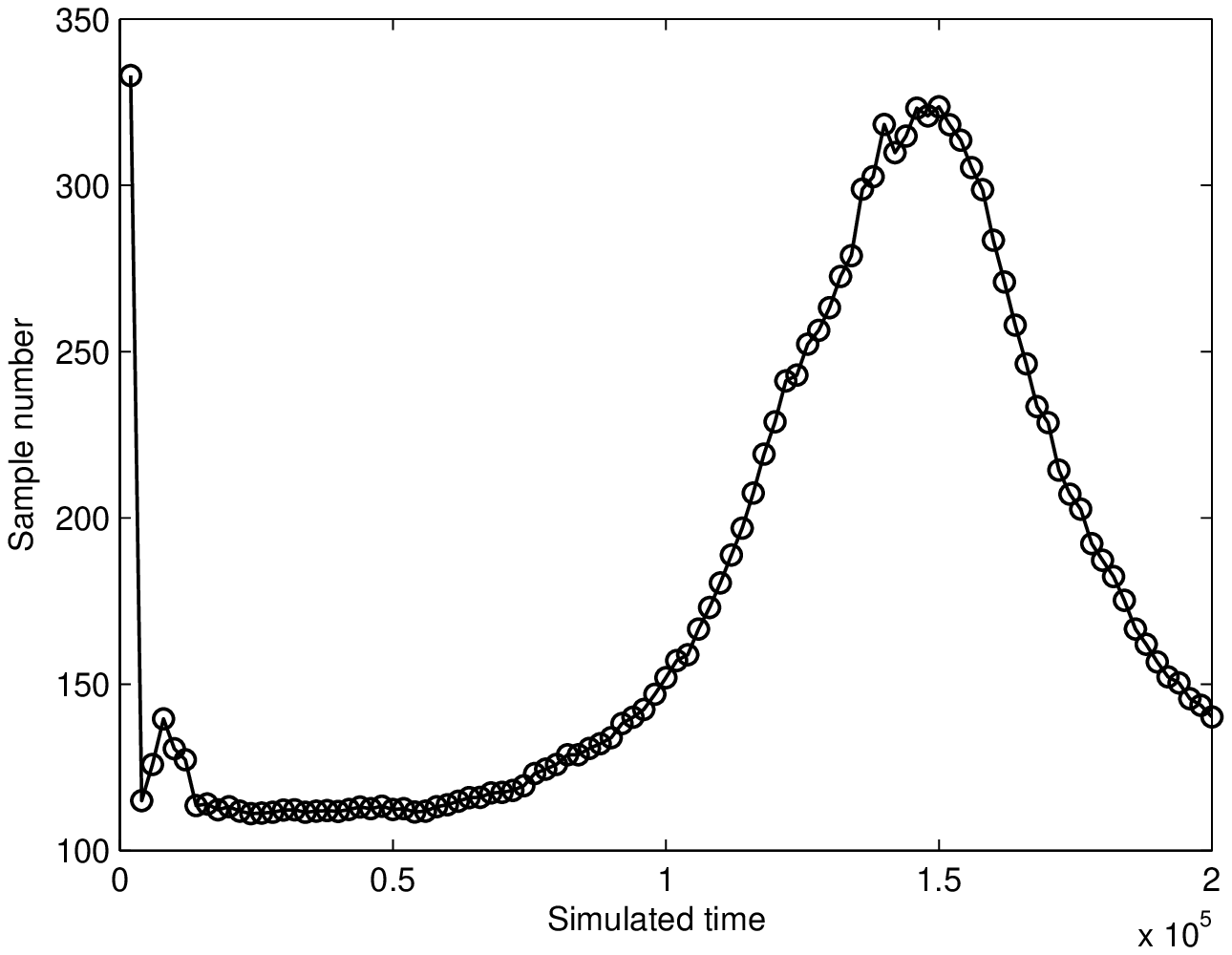}}
  \centering
\parbox{9cm}{\caption{Corresponding number of samples in each estimation window}\label{track}}
\end{figure}

%% file: cs.tex
\section{A Discussion on the Applicability of Compressive Sensing}\label{sec:CS}

Recent advances in compressive sensing theory \cite{cs1},\cite{cs2},\cite{cs3} allow one to represent compressible/sparse signals with significantly fewer samples than required by the Nyquist sampling theorem.  It is therefore particularly attractive in a resource constrained setting.  This technique has been used in data compression \cite{cs_compression}, channel coding \cite{cs_channelcoding}, analog signal sensing \cite{cs_physical}, routing \cite{ITA} and data collection \cite{cs_datacollection}. % in a sensor network.
It is tempting to examine whether this technique brings any advantage for %could be used in
our channel estimation problem.  The idea is to randomly sample the channel state, use compressive sensing techniques to reconstruct the entire sequence of channel state evolution, and then use the ML estimator to determine the channel parameter.  Compared to the sensing schemes discussed in the previous sections, this is an {\em indirect} use of the ML estimator, in that the entire sequence will be reconstructed before the estimation.  In this sense the use of compressive sensing also seems to be an overkill for the purpose of parameter estimation.
%As we shall see here our results are not as conclusive as one hoped, though they do suggest an interesting direction of future research.

%We wonder whether compressive sensing can be applied in our channel estimation scenario. If it can reconstruct the whole channel state sequence with a few samples, then better estimation performance can be obtained compare to that without using compressive sensing.

%\subsection{Brief overview of compressive sensing}

Consider a vector of discrete-time, finite, one-dimensional signal $\mathbf{x}_{N\times 1}$, which can be expressed as $\mathbf{x}=\mathbf{\Psi}\mathbf{a}$, where $\mathbf{\Psi}$ is an $N\times N$ basis matrix and $\mathbf{a}$ is a vector of weighting coefficients.  The signal vector $\mathbf{x}$ is $K$-sparse if $\mathbf{a}_{N\times1}$ has only $K$ non-zero elements.
The compressive sensing theory states that the signal $\mathbf{x}$ can be reconstructed successfully by $M$ measurements $\mathbf{y}$, which is done by projecting the signal $\mathbf{x}$ to another basis $\mathbf{\Phi}$ that is incoherent with $\mathbf{\Psi}$, i.e., $\mathbf{y}=\mathbf{\Phi}\mathbf{x}=\mathbf{\Phi}\mathbf{\Psi}\mathbf{a}$.  %The precise definition of incoherence is given later.
The required length of $\mathbf{y}$, $M$, depends on the sparsity of the signal and the reconstruction algorithm. The reconstruction is typically done by solving the $l_{1}$\emph{-norm} optimization problem: $\hat{\mathbf{a}}=\arg\min\|\mathbf{a\|_{1}}$, s.t. $\mathbf{y}=\mathbf{\Phi}\mathbf{\Psi}\mathbf{a}$.  Algorithmically this can be solved by linear programming or iterative greedy algorithm such as
orthogonal matching pursuit (OMP) \cite{OMP}.

%From the study of the existing research work, we find that the practicality of compressive sensing is not considered circumspectly. Firstly, the measurement matrix $\mathbf{\Phi}$ mostly has no physical sense so that with less pertinence to the associated research scenario. Secondly, to apply compressive sensing algorithm some additional device should be equipped to reflect the original data to fewer measurements.

%\subsection{Compressive sensing based channel estimation}

%To the best of our knowledge, there has not been studies on applying compressive sensing (CS) in the above channel estimation scenario. The closest work we found is \cite{ITA}, where the authors studied a problem of joint routing and compression: measurements are taking by a subset of nodes in a network and routed to a central destination, where they are used to reconstruct the original signal.  As measurements are delivered along a route, they are linearly combined, thus in this context matrix $\mathbf{\Phi}$ is determined by the routing topology and not randomly generated as is typically the case in the CS literature.

%of reconstructing the original signal from a small subset of sample is essentially similar to our problem. The authors address the joint routing and compression problem where the matrix $\mathbf{\Phi}$ is built on the fly according to the routing topology.
%Our problem
%has a similar feature and is even more
%is rather unusual compared to many other scenarios where compressive sensing has been applied to.  %in that our signal (channel states) is directly measured.
For our channel estimation problem, consider the signal $\mathbf{x}=\{ x_1, x_2, \cdots, x_N\}$ to be the discrete time 0-1 sequence of channel states, with $x_i$ denoting the channel state at time $i$.  The physical nature of channel sensing implies the measurement matrix $\mathbf{\Phi}_{M\times N}$ consists of rows each containing only a single $1$ in the position where the channel was sensed and $0$ everywhere else.  Specifically, a $1$ in the position $(i, j)$ means that the $i$th measurement was taken at time $j$.  In addition, there can only be one measurement taken at time $j$, i.e., no two rows can have a $1$ in the same column.   As $M < N$ in general (or it wouldn't be {\em compressive} sensing), there will be exactly $N-M$ empty (all-$0$) columns, making the matrix extremely sparse.  This poses a significant challenge since in general the $\mathbf{\Phi}$ matrix is required to be dense (though randomly generated), with at least one non-zero entry in each column. % (this is so that each element of the signal is represented/included in some measurement/linear combination).

For the reconstruction to be successful, two conditions need to be satisfied: the signal needs be sparse in some domain (i.e., the existence of a $\mathbf{\Psi}$ such that $\mathbf{a}$ is sufficiently sparse), and the two matrices $\mathbf{\Phi}$ and $\mathbf{\Psi}$ need to be incoherent.  Due to the binary property of the channel state sequence, it's difficult to find a basis matrix $\mathbf{\Psi}$ that has dense entities.
%
%This means that none of the elements of any column in one matrix has a sparse representation in terms of the columns of the other matrix.
%We measure the incoherence of two matrices as follows \cite{ITA}.  Projecting each row of $\Phi$ onto the space spanned by the columns of $\Psi$ we get:
%\begin{equation}
%\zeta_{j}=(\Psi^{T}\Psi)^{-1}\Psi^{T}\phi_{j}^{T} ~,
%\end{equation}
%where $\phi_{j}$ is the $j$th row of $\Phi$ and $\zeta_{j}$ is the vector of coefficients corresponding to its projection on the space spanned by the columns of $\Psi$. A measure for the incoherence is then defined as
%\begin{equation}
%I(\Phi,\Psi)=\min_{j=1,\dots,N}\big [\sum_{i=1}^{N} 1 \{ \rho_{i}^{j} \neq 0 \} \big ] \in[1,N],
%\end{equation}
%where $\rho_{i}^{j}$ is the $i$th entry of vector $\zeta_{j}$ and $1 \{E\}$ is the indicator function, which is $1$ whenever event $E$ is true and $0$ otherwise. The larger this quantity, the more incoherent the two matrices.
%
%The challenge we face is that
%In addition, the choices for the basis matrix and the physical nature of the measurement matrix result in
As a result we have two very sparse matrices and they are highly coherent.  %This does not bode well for its performance.
For these reasons we have not found compressive sensing to have an advantage in our channel estimation problem.
%
%% between these matrices.
%%Put differently, the coherence metric measures the largest correlation between any two elements of $\mathbf{\Phi}$ and $\mathbf{\Psi}$, and can be defined as:
%%\begin{displaymath}
%%\mu(\mathbf{\Phi},\mathbf{\Psi})=\sqrt{N}\max_{1\leq i,j\leq N}|\langle\phi_{i},\psi_{j}\rangle|.
%%\end{displaymath}
%%The smaller the coherence between sampling matrix and basis matrix is, the less measurements are needed to reconstruct the signal.

%Take a $128\times128$ basis matrix as an example, the incoherences of  the sampling matrix and $\mathbf{I}$ is $1$,  the one of the sampling matrix and Harr basis is $2$. We can find that the sampling matrix is highly coherent with the basis matrix. This is mainly because the the sparsity of both the sampling matrix and the basis matrix. Therefore, it is difficult to reconstruct the signal with a few measurements.

Figure \ref{com1} shows some comparison results.  In the simulation of compressive sensing based estimation, we reconstruct the original state sequence using Harr wavelet basis. All other conditions remain the same as in previous sections.
%The performance comparison is shown in Figure \ref{com1} with different pairs of channel parameters where
The time window is set to  $4096$ time units.
Overall compressive sensing based estimation dose not compare favorably with uniform sensing and random sensing, due to the coherence problem between the two matrices. %    This is mainly because of the error in the reconstructed sequence due to the coherence between the measurement and the basis matrices.
It remains an interesting problem to find a good basis matrix that can both sparsify $\mathbf{x}$ and at the same time be sufficiently incoherent with the measurement matrix.  A similar difficulty was noted in \cite{ITA} in trying to use compressive sensing for a data gathering problem.  A number of commonly used transformations were considered, and it was found that, with real data sets, none of them was able to sparsify the data while being at the same time incoherent with the routing matrix. %The resulting performance was thus not as good as expected.

% This is so even in the case of few on/off transitions (e.g., long off periods and short on periods in the last figure) where one might have expected it to perform well.
%Specifically,  in the last pair of figures the channel experienced few on-off transitions as a result of (1) long off periods and short on periods, (2) long on periods and long off periods, respectively. We expect compressive sensing based estimation could perform better than uniform sampling and random sampling based in these two scenarios as mentioned above. However the results are not satisfactory.
%These results are a manifestation of the coherence difficulty mentioned above.  We would need many samples to reconstruct the channel sequence due to the coherence between the measurement matrix and basis matrix. %, otherwise the reconstruction is unsatisfactory. }

\begin{figure}%[!t]x
  % Requires \usepackage{graphicx}
 \subfigure[{$E[T_{0}]=2$, $E[T_{1}]=1$}] {\includegraphics[width=0.32\textwidth]{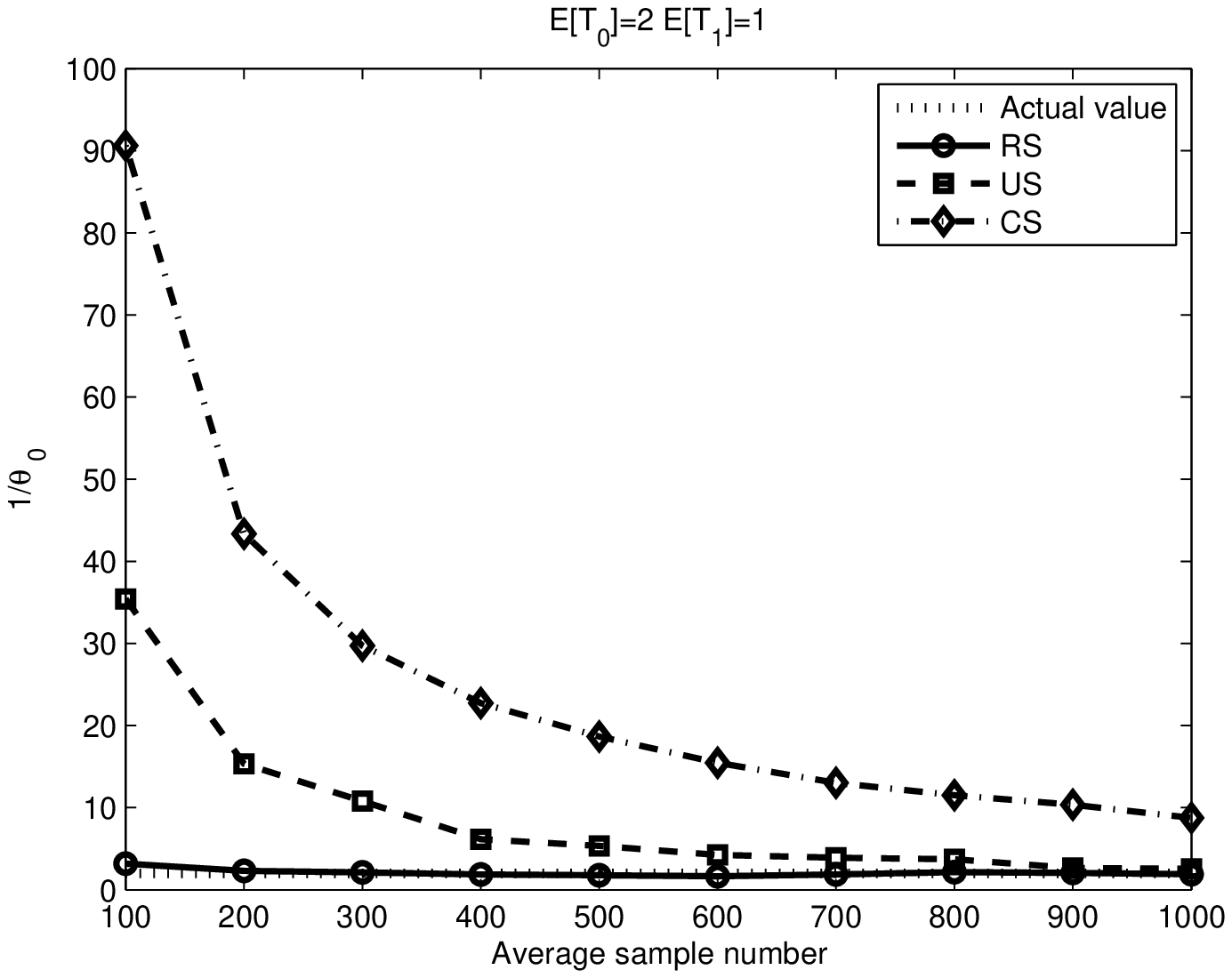}}
 \subfigure[{$E[T_{0}]=20$, $E[T_{1}]=10$}] {\includegraphics[width=0.32\textwidth]{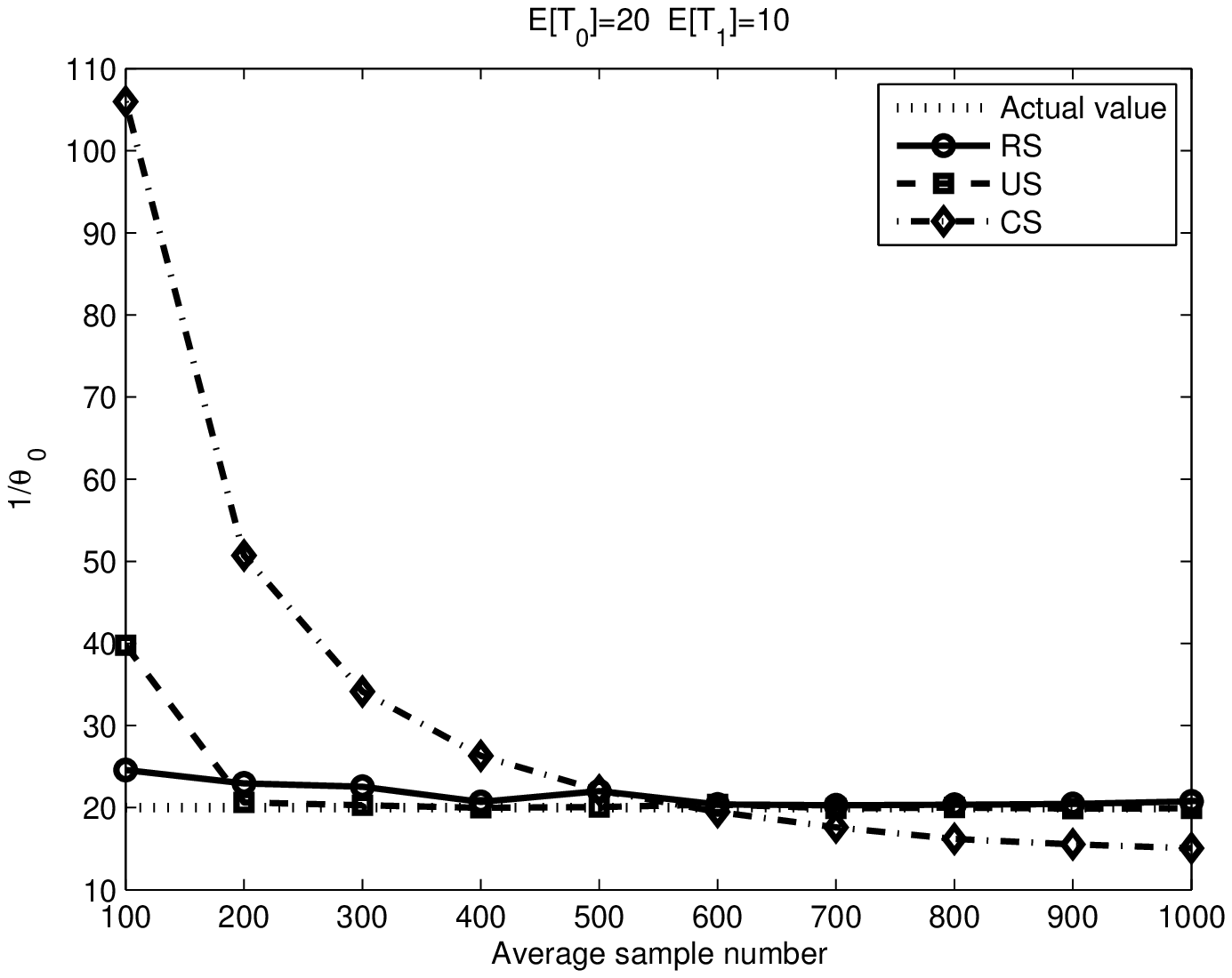}}
 \subfigure[{$E[T_{0}]=100$, $E[T_{1}]=5$}] {\includegraphics[width=0.32\textwidth]{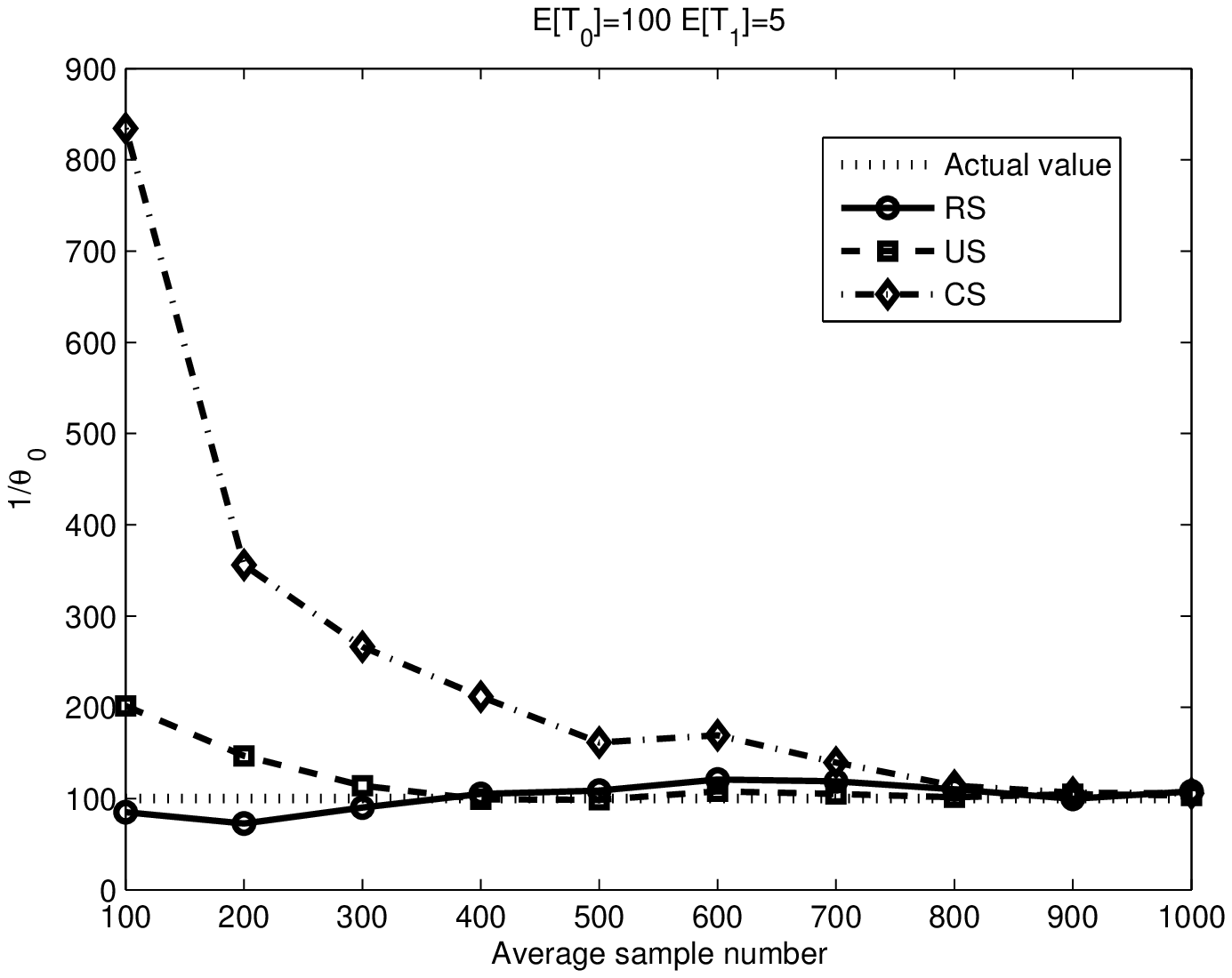}}
  \centering
  \parbox{14cm}{\caption{Estimation performance comparison: random sensing vs. uniform sensing vs. CS based sensing }\label{com1}}
  \end{figure}

%% file: proof.tex
\section{Proof of Lemma 1}
\begin{proof}
For simplicity in presentation, we first write $g(\Delta t)=h_o(\Delta t)h(\Delta t)$, where
\begin{eqnarray*}
h_o(\Delta t) &=&\frac{\Delta t^{2}}{u^{2}}e^{-\theta_{0}\Delta t/u},\\
h(\Delta t) &=&
h_1(\Delta t) + h_2(\Delta t) + h_3(\Delta t) ~,
\end{eqnarray*}
where
\begin{eqnarray*}
h_1(\Delta t)&=&\frac{2u(1-u)}{1-e^{-\theta_{0}\Delta t/u}}, \nonumber\\
h_2(\Delta t)&=&-\frac{u(1-u)^{2}}{(1-u)+ue^{-\theta_{0}\Delta t/u}} ,\nonumber \\
h_3(\Delta t)&=&-\frac{u^{2}(1-u)}{u+(1-u)e^{-\theta_{0}\Delta t/u}}.\\
\end{eqnarray*}

We proceed to show that each of the above functions is convex under Condition \ref{cond:sparsity}.

We first show that $h_o(\Delta t)$ is strictly convex for $\Delta t > (2+\sqrt{2})u/\theta_0$.
Under this condition and noting $0<u<1$ and $\theta_0>0$ we have
%Given $f(\Delta t)=\frac{\Delta t^{2}}{u^{2}}e^{-\theta_{0}\Delta t/u}$, we have
\begin{eqnarray*}
h_o^{'}(\Delta t)&=&\frac{\Delta t}{u^2}e^{-\theta_{0}\Delta t/u}(2-\frac{\theta_0 \Delta t }{u}) <0,\\
h_o^{''}(\Delta t)&=&\frac{e^{-\theta_{0}\Delta t/u}}{u^2}[(\frac{\theta_0 \Delta t }{u}-2)^2-2] >0 .\\
\end{eqnarray*}
Therefore for $\frac{\theta_0 \Delta t}{u} > 2+\sqrt{2}$, $h_o(\Delta t)$ is strictly convex.  %Furthermore, since $f^{'}(\Delta
That $h_1(\Delta t)$ is strictly convex is straightforward.  Since $0<u<1$ and $\theta_0>0$, we have:
\begin{eqnarray*}
h_1^{'}(\Delta t)&=&\frac{-2(1-u)\theta_0 e^{-\theta_{0}\Delta t/u}}{(1-e^{-\theta_{0}\Delta t/u})^2} <0 ,\\
h_1^{''}(\Delta t)&=&\frac{2(1-u)\theta_0^{2}e^{-\theta_{0}\Delta t/u}(1+e^{-\theta_{0}\Delta t/u})}{u(1-e^{-\theta_{0}\Delta t/u})^3}>0.\\
\end{eqnarray*}

Next we show that $h_2(\Delta t)$ is strictly convex for $\Delta t >  \frac{u}{\theta_0} \ln(\frac{u}{1-u})$.  This condition is equivalent to $ue^{-\theta_0 \Delta t/u} < 1-u$.  Under this condition and again noting $0<u<1$ and $\theta_0>0$, we have
\begin{eqnarray*}
h_2^{'}(\Delta t)&=&\frac{-u(1-u)^{2}\theta_0 e^{-\theta_{0}\Delta t/u}}{[(1-u)+ue^{-\theta_{0}\Delta t/u}]^2} < 0,\\
h_2^{''}(\Delta t)&=&\frac{(1-u)^2\theta_0^{2}e^{-\theta_{0}\Delta t/u}[(1-u)-ue^{-\theta_{0}\Delta t/u}]}{[(1-u)+ue^{-\theta_{0}\Delta t/u}]^3} > 0.\\
\end{eqnarray*}

Similarly, $h_3(\Delta t)$ is strictly convex under the condition $\Delta t >  \frac{u}{\theta_0} \ln(\frac{1-u}{u})$, since
\begin{eqnarray*}
h_3^{'}(\Delta t)&=&\frac{-u(1-u)^{2}\theta e^{-\theta_{0}\Delta t/u}}{[u+(1-u)e^{-\theta_{0}\Delta t/u}]^2} <0,\\
h_3^{''}(\Delta t)&=&\frac{(1-u)^2\theta^{2}e^{-\theta_{0}\Delta t/u}[u-(1-u)e^{-\theta_{0}\Delta t/u}]}{[u+(1-u)e^{-\theta_{0}\Delta t/u}]^3}>0.\\
\end{eqnarray*}

Therefore under the condition $\Delta t>\alpha u/ \theta_0$, %$\alpha = \max\{2+\sqrt{2}, \ln(\frac{1-u}{u}), \ln(\frac{u}{1-u})\}$, $f$,
$h_1$, $h_2$ and $h_3$ are all monotonically decreasing convex functions.  It follows that $h=h_1+h_2+h_3$ is also monotonically decreasing and convex.  Furthermore, for any $\Delta t >0$, $h_o(\Delta t)>0$, and $h(\Delta t) > h(+\infty) = 0$.  We can now show that $g$ is strictly convex under this condition:
\begin{equation}
\label{g''}
\begin{split}
g^{''}(\Delta t)&=(h_o(\Delta t)h(\Delta t))^{''}\\
&=h_o^{''}(\Delta t)h(\Delta t)+2h_o^{'}(\Delta t)h^{'}(\Delta t)+h_o(\Delta t)h^{''}(\Delta t)>0 ~,
\end{split}
\end{equation}
where the inequality holds because every term on the right hand side is positive under the condition
$\Delta t>\alpha u/ \theta_0$ as summarized above.
\end{proof}

%% file: draft_jour.bbl
\begin{thebibliography}{1}

\bibitem{Haykin}
S. Haykin, ``Cognitive radio: brain-empowered wireless
communications,'' \emph{ IEEE journal on selected areas in
communications}, vol. 23, pp: 201-220, Feb. 2005.
\bibitem{Challapali}
K. Challapali, C. Cordeiro, and D. Birru, ``Evolution of
Spectrum-Agile Cognitive Radios: First Wireless Internet Standard
and Beyond,'' \emph{Proceedings of ACM International Wireless
Internet Conference}, August 2006.
\bibitem{Akyildiz}
I. F. Akyildiz, W.-Y. Lee, M. C. Vuran, and S. Mohanty, ``NeXt
generation dynamic spectrum access cognitive radio wireless
networks: A survey,'' \emph{Computer Networks Journal (Elsevier)},
pp: 201-220, Sept. 2006.
\bibitem{Cabric}
D. Cabric, S. M. Mishra, R. W. Brodersen, ``Implementation issues in spectrum sensing for cognitive
radios,'' \emph{Proceedings of Asilomar Conference on Signals, Systems and Computers}, 2004.
\bibitem{kim1}
H. Kim and K. G. Shin, ``Efficient discovery of spectrum
opportunities with MAC-layer sensing in cognitive radio networks,''
\emph{IEEE Transactions on Mobile Computing},
vol.7, no.5, pp: 533-545, May 2008.
\bibitem{Zhao}
Q. Zhao, L. Tong, A. Swami, and Y. Chen, ``Decentralized cognitive
MAC for opportunistic spectrum access in ad hoc networks: a POMDP
framework,'' \emph{IEEE Journal on Selected Areas in
Communications}, vol. 25, no. 3, pp: 589-599, Apr. 2007.
\bibitem{Tong}
M. Dong, L. Tong, and B. M. Sadler, ``Information Retrieval and Processing in Sensor Networks: Deterministic Scheduling Versus Random Access,`` \emph{IEEE Transactions on Signal Processiong}, vol. 55, no. 12, pp: 5806-5820, 2007.
\bibitem{kim2}
H. Kim and K. G. Shin, ``Fast Discovery of Spectrum Opportunities
in Cognitive Radio Networks,'' \emph{Proceedings of the 3rd IEEE
Symposia on New Frontiers in Dynamic Spectrum Access Networks
(IEEE DySPAN)}, Oct. 2008.
\bibitem{Crowncom}
X. Long,   X. Gan,   Y. Xu,   J. Liu,   M. Tao, ``An Estimation
Algorithm of Channel State Transition Probabilities for Cognitive
Radio Systems,'' \emph{Proceedings of Cognitive Radio Oriented
Wireless Networks and Communications (CrownCom)}, 15-17 May 2008
\bibitem{HMM}
C. H. Park, S. W. Kim,  S. M. Lim,  M. S. Song, ``HMM Based
Channel Status Predictor for Cognitive Radio,''  \emph{Proceedings
of Asia-Pacific Microwave Conference},11-14 Dec. 2007.
\bibitem{Wavelet}
A. A. Fuqaha,  B. Khan, A.  Rayes,  M. Guizani,  O. Awwad,  G. Ben
Brahim, ``Opportunistic Channel Selection Strategy for Better QoS
in Cooperative Networks with Cognitive Radio Capabilities,''
\emph{IEEE Journal on Selected Areas in Communications}, Vol. 26,
No. 1, pp: 156-167, Jan. 2008.
\bibitem{mobicom09}
D. Chen, S. Yin, Q. Zhang, M. Liu and S. Li, ``Mining Spectrum Usage Data: a Large-scale Spectrum Measurement Study,'' {\em ACM MobiCom}, September 2009, Beijing, China.
\bibitem{underlay}
P. J. Kolodzy, "Cognitive radio fundamentals," \emph{Proceedings of  SDR Forum}, Singapore, Apr. 2005.
\bibitem{ML_defination}
R. A. Fisher, ``On the Mathematical Foundations of Theoretical Statistics,''
\emph{Mathematical Foundations of Theoretical Statistics}
vol. 222, pp: 309-368, 1922.
\bibitem{renewal}
D. R. Cox, \emph{Renewal Theory}, Butler and Tanner,1967.
\bibitem{ML1}
J. Lee, J. Choi, H. Lou, ``Joint Maximum Likelihood Estimation of Channel and Preamble Sequence for WiMAX Systems,''
\emph{IEEE Transactions on Wireless Communications},
Vol. 7,  No. 11, pp: 4294-4303, Nov. 2008.
\bibitem{ML2}
J. Wang, A. Dogandzic, A. Nehorai, ``Maximum Likelihood Estimation of Compound-Gaussian Clutter and Target Parameters,''
\emph{IEEE Transactions on Signal Processing},
vol. 54,  no.10,  pp: 3884-3898  Oct. 2006.
\bibitem{ML3}
U. Orguner, M. Demirekler, ``Maximum Likelihood Estimation of Transition Probabilities of Jump Markov Linear Systems,'' \emph{IEEE Transactions on Signal Processing},
vol. 56,  no. 10,  Part 2, pp: 5093-5108  Oct. 2008.
\bibitem{MLchannel1}
H.A. Cirpan, M.K. Tsatsanis, ``Maximum likelihood blind channel estimation in the presence of Doppler shifts,''
\emph{ IEEE Transactions on Signal Processing},
vol. 47,  no. 6, pp: 1559-1569, Jun. 1999.
\bibitem{MLchannel2}
M. Abuthinien, S. Chen, L. Hanzo, ``Semi-blind Joint Maximum Likelihood Channel Estimation and Data Detection for MIMO Systems,''\emph{ IEEE Signal Processing Letters}, vol. 15, pp: 202-205, 2008.
\bibitem{ML}
A.W. van der Vaart, \emph{Asymptotic Statistics} (Cambridge Series in Statistical and Probabilistic Mathematics) (1998)
\bibitem{fisher}
R. A. Fisher, \emph{The Design of Experiments}, Oliver and Boyd, Edinburgh, 1935.
\bibitem{FisherInfo}
J.A. Legg, D.A. Gray, ``Performance Bounds for Polynomial Phase Parameter,'' \emph{IEEE Transactions on Signal Processing},
vol. 48,  no.2,  pp: 331-337  Feb. 2000.
\bibitem{Circular}
M. L. Mehta,  \emph{Random matrices}. Elsevier/Academic Press, 2004.
\bibitem{Jacobi}
A. Edelman, B.D. Sutton, ``The Beta-Jacobi Matrix Model, the CS Decomposition, and Generalized Singular Value Problems,'' \emph{Foundations of Computational Mathematics}, vol. 8 ,  no. 2, pp: 259-285, May 2008.
\bibitem{cs1}
D. Donoho, ``Compressed sensing,'' \emph{IEEE Transactions on Information
Theory}, vol. 52, no. 4, pp: 4036-4048, 2006.
\bibitem{cs2}
E. Cand\'{e}s and T. Tao, ``Near optimal signal recovery from random
projections: Universal encoding strategies?'' \emph{IEEE Transactions on
Information Theory}, vol. 52, no. 12, pp: 5406-5425, 2006.
\bibitem{cs3}
E. Cand\'{e}s, J. Romberg, and T. Tao, ``Robust uncertainty principles:
Exact signal reconstruction from highly incomplete frequency information,''
\emph{IEEE Transactions on Information Theory}, vol. 52, no. 2, pp: 489-509, 2006.
\bibitem{cs_compression}
D. Baron, M.B. Wakin, M.F. Duarte, S. Sarvotham, and R.G. Baraniuk,
``Distributed compressed sensing,'' 2005, Preprint.
\bibitem{cs_channelcoding}
E. Cand\'{e}s and T. Tao, ``Decoding by linear programming,'' \emph{IEEE Transactions
Information Theory}, vol. 51, no. 12, pp: 4203-4215, Dec. 2005.
\bibitem{cs_physical}
Z. Tian and G. Giannakis, ``Compressed Sensing for Wideband Cognitive Radios,'' \emph{Proceedings of IEEE Internation Conference on Acoustics, Speech and Signal Processing (ICASSP)}, Vol. 4, pp: 1357-1360, Honolulu, Apr. 2007.
\bibitem{ITA}
G. Quer, R. Masiero, D. Munaretto, M. Rossi, J. Widmer and M. Zorzi, ``On the Interplay Between Routing and Signal Representation for Compressive Sensing in Wireless Sensor Networks,''  \emph{Information Theory and Applications Workshop (ITA 2009)}, San Diego, CA.
\bibitem{cs_datacollection}
C.  Luo, F. Wu, C. W. Chen and J. Sun,
``Compressive Data Gathering for Large-Scale Wireless Sensor Networks,'' {\em ACM MobiCom}, September 2009, Beijing, China.
\bibitem{OMP}
J. A. Tropp, A. C. Gilbert, ``Signal Recovery From Random Measurements Via Orthogonal Matching Pursuit,'' \emph{IEEE Transactions on Information Theory}, vol. 53, no. 12, pp: 4655-4666, 2007.
\end{thebibliography}
